\documentclass[%
 reprint,
 amsmath,amssymb,
 aps,
]{revtex4-2}

\usepackage{graphicx}
\usepackage{dcolumn}
\usepackage{bm}
\usepackage{physics}

\graphicspath{{images/}}
\usepackage{amssymb}
\usepackage{amsmath}
\usepackage{tikz}
\usepackage{blochsphere}
\usepackage{hyperref}
\usepackage{amsthm}
\usepackage{mathrsfs}
\usepackage{xr}
\usepackage{color,soul} 
\definecolor{myellow}{RGB}{220,220,0}

\usepackage{tikz-cd}
\usetikzlibrary {shapes.geometric}

\theoremstyle{plain}
\newtheorem{thm}{Theorem}[section]

\theoremstyle{definition}

\newtheorem{prop}[thm]{Proposition}

\theoremstyle{remark}
\newtheorem*{rem}{Remark}

\newtheoremstyle{mydef}
	{\topsep}   
    {\topsep}   
    {}  
    {0pt}       
    {\bfseries} 
    {.}         
    {5pt plus 1pt minus 1pt} 
    {}          

\theoremstyle{mydef}
\newtheorem{definition}[thm]{Definition}

\tikzset{->-/.style={decoration={
  markings,
  mark=at position .5 with {\arrow{>}}},postaction={decorate}}}

\begin{document}

\preprint{APS/123-QED}

\title{Sequential topology: iterative topological phase transitions in finite chiral structures}

\author{Maxine M. McCarthy}
\email[Corresponding author: ]{maxine.mccarthy@mpl.mpg.de}
\affiliation{%
 Physics and Astronomy, School of Mathematical and Physical Sciences, University of Sheffield, Sheffield S3 7RH, UK
}%
\affiliation{%
 Max Planck Institute for the Science of Light, Staudstraße 2, 91058 Erlangen, Germany
}%

\author{D. M. Whittaker}
\affiliation{%
 Physics and Astronomy, School of Mathematical and Physical Sciences, University of Sheffield, Sheffield S3 7RH, UK
}%

\date{May 15, 2025}

\begin{abstract}
We present theoretical and experimental results probing the rich topological structure of arbitrarily disordered finite tight binding Hamiltonians with chiral symmetry. We extend the known classification by considering the topological properties of phase boundaries themselves. That is, can Hamiltonians that are confined to being topologically marginal, also have distinct topological phases? For chiral structures, we answer this in the affirmative, where we define topological phase boundaries as having an unavoidable increase in the degeneracy of real space zero modes. By iterating this question, and considering how to enforce a Hamiltonian to a phase boundary, we give a protocol to find the largest dimension subspace of a disordered parameter space that has a certain order degeneracy of zero energy states, which we call \textit{sequential topology}. We show such degeneracy alters localisation and transport properties of zero modes, allowing us to experimentally corroborate our theory using a state-of-the-art coaxial cable platform. Our theory applies to systems with an arbitrary underlying connectivity or disorder, and so can be calculated for any finite chiral structure. Technological and theoretical applications of our work are discussed.
\end{abstract}

\maketitle

\section{Introduction}\label{Intro}

\noindent Topology is a pillar of modern condensed matter physics, where novel and unusual disorder-resistant behaviour can be associated to a system with topologically non-trivial properties. For example, boundary localised states \cite{ExpLoc}, that can exhibit scatter free directional transport \cite{QH1,QH2,QH3}, and are anticipated to be non-trivial in up to 90\% of materials \cite{TQC1,TQC2,TQC3}. Typically such phenomena are a consequence of a Hamiltonian with a gapped bulk having gapless boundary states. The presence and behaviour of such states are well predicted with the bulk-boundary correspondence (BBC) \cite{10foldway,FragileTopology,ExpLoc,Ktheory}. \\ 
\indent The BBC applies exactly only in the (infinite) thermodynamic limit. In a finite system topological phenomena converge to their predicted behaviour exponentially with system size, so that there exists (an exponentially vanishing \cite{ExpConvergence,NonComTorus,RobustnessThermoLim,NonComIndices}) finite disorder at which the BBC is effectively lost. That is, gapless modes may become gapped and lose their associated behaviour. Such loss of the BBC is most prevalent in very small (10s to 100s of sites) structures with very strong disorder, or those with a highly random network topology. Although the BBC is well applied to such systems at low disorder, many cutting edge experimental platforms are small enough for this regime to be of importance, for instance optomechanical arrays \cite{OptomechanicalCombDevice}, micropillar polariton devices \cite{polariton1,polariton2}, acoustic metamaterials \cite{acoustics1,acoustics2}, mechanical systems \cite{Mechanical1}, topoelectric circuits \cite{Topoelectic0,Topoelectric1,Topoelectric2}, and coaxial cable networks \cite{cables1,cables2,cables3}.

\indent We have previously explored one approach to studying topological phases in finite chiral structures at arbitrarily strong disorder \cite{classification1}, for which an extension beyond chiral symmetry is the subject of upcoming work. 
Under our approach two (real-space) Hamiltonians are considered topologically distinct if and only if any continuous evolution of hopping terms that maps between them involves a gap closure in the energy spectrum. Our previous approach assumes algebraic independence of all hopping terms, so that any disordered hopping distribution of a particular Hamiltonian can be studied. Such a Hamiltonian is confined to live in a parameter space $\xi$ given by the tuple of each hopping amplitude $\xi=(u,v,\cdots)$, with an example graphically displayed in Fig. \ref{FiniteZerothStepTopology}. 
However if hopping terms are dependant on one another, this approach has the potential to miss important physical behaviour. For example topological bubbles \cite{FiniteSiteTop}, convergence of edge states \cite{Realspace1,Pseudospectra,Realspace3,RobustnessThermoLim}, 
and chiral transport \cite{Realspace2} 
predicted by other approaches to the classification of finite media, as well as the localisation-delocalisation transitions we discuss presently.

\indent Starting from a tight binding Hamiltonian with all independent hopping terms, in this paper we study how the topological classification of a structure is altered as we begin to introduce algebraic dependence. This allows us to constrain Hamiltonians to a topological phase boundary, and ask if that phase boundary itself has a topologically non-trivial classification. This question can be iterated by increasing the number of algebraic constraints on a Hamiltonian. Mathematically this corresponds to composing functions applied to the hopping terms of the Hamiltonian. In this setting, we define topological phase boundaries as those where there is an unavoidable increase in the degeneracy of zero energy states (for instance from a $2n$ degeneracy to a $2n+2$ degeneracy). The order of degeneracy of a phase boundary then increases with the number of applied constraints. Formally, sequential topology classifies constrained subspaces of the parameter space $\xi$. \\
\indent Herein we focus on topological properties of phase boundaries as most other choices of constraints will reveal the same classification as in Ref. \cite{classification1}. Despite our choice, we do note that other topologically interesting choices do exist, although alternative choices of constraints are beyond the scope of this work.

We found that applying a sequence of constraints reveals a rich classification, beyond the starting point with all independent hopping terms, which we call \textit{sequential topology}. We focus on systems with chiral symmetry in this paper so Hamiltonians have an energy spectrum which is symmetric around zero energy, however the methods we describe are readily generalisable to systems with other symmetries.

\begin{figure}
\centering
\begin{tikzpicture}
\node (0) at (0,0) {\includegraphics[width=0.8\linewidth]{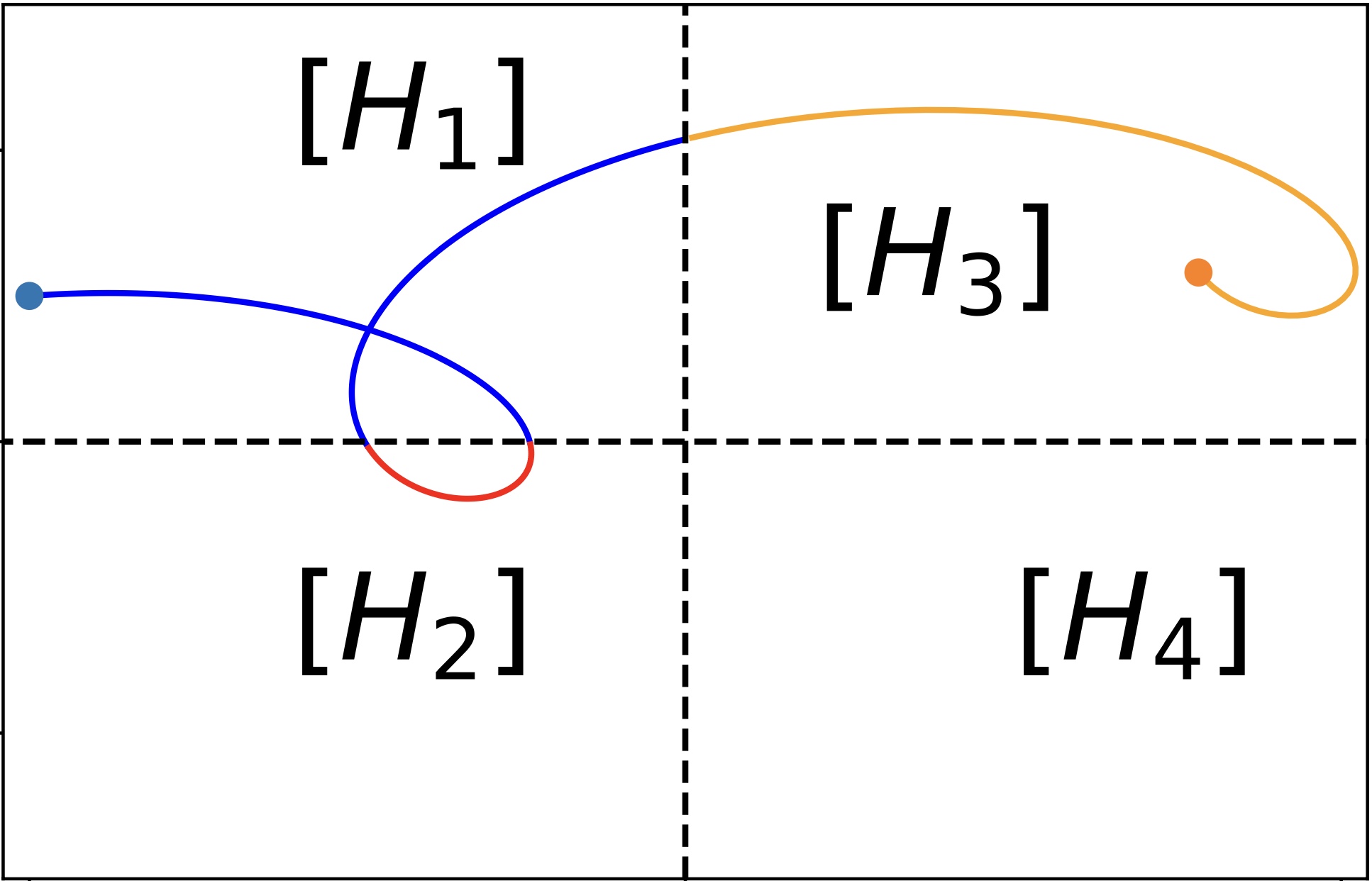}};
\node (a) at (0,-2.5) {\Large $a$};
\node[rotate=90] (b) at (-3.8,0) {\Large $b$};
\end{tikzpicture}
\caption{A slice of the parameter space for a structure with the zeroth step classification $2\mathbb{Z}_2$. This 
slice is defined by keeping all the hopping terms constant, apart from one in each
section,  
$a$ and $b$, which are allowed to evolve independently. 
$[H_x]$ denotes a set of topologically equivalent Hamiltonians, and the dashed 
lines denote phase boundaries. The path shown  undergoes phase transitions between 
the phases $[H_1]$, $[H_2]$, and $[H_3]$, where each phase transition is identified
by the appearance of a pair of unavoidable degenerate zero energy states.}
\label{FiniteZerothStepTopology}
\end{figure}

\indent Remarkably we find for many real finite chiral Hamiltonians, there is a near universal sequence of topological classifications when composing a sequence of algebraic constraints on the Hamiltonian. We choose these constraints in a way that minimise the necessary number to reveal a new topological classification. This may be represented as a sequence which we call the \textit{sequential classification},
\begin{widetext}
\begin{equation}
\begin{tikzcd}
N\mathbb{Z}_2 \arrow{r} & 0 \arrow{r} & M\mathbb{Z}_2 \arrow{r} & 0 \arrow{r} & 0 \arrow{r} & P\mathbb{Z}_2 \arrow{r} & 0 \arrow{r} & 0 \arrow{r} & 0 \arrow{r} & R\mathbb{Z}_2 \arrow{r} & 0 \arrow{r} & \cdots
\end{tikzcd}
\end{equation}
\end{widetext}
where each horizontal arrow denotes the addition of a constraint on the hopping terms. \\
\indent When the application of a number of constraints leads to a non-trivial classification, this is a consequence of the constrained subspace having multiple topological phases. Each phase is separated by a boundary with an increased degeneracy in zero energy states. A graphical example of constrained subspaces being topologically non-trivial is displayed in Fig. \ref{CartoonPic2}. It is possible to observe a sequential phase transition as the newly degenerate zero energy states undergo an unusual delocalisation at the phase boundary. This is experimentally observable by mapping out the local density of states (LDOS) or by measuring the transport properties of the system. Using these physical consequences, we experimentally corroborate our sequential classification for the structure in Fig. \ref{ExpStruc}. 
This also results in a number of possible technological applications of sequential topology as it allows for the precise control of the localisation and scattering properties of a system while only needing to maintain a minimum number of controlled parameters.

\begin{figure}
\begin{tikzpicture}[scale=0.75]
\draw[fill=purple,fill opacity=0.2] (0,0) -- (0,5) -- (3,7) -- (3,2) -- (0,0);

\draw[dashed] (0+4,0) -- (0+4,5) -- (3+4,7) -- (3+4,2) -- (0+4,0);

\draw[dashed] (0-4,0) -- (0-4,5) -- (3-4,7) -- (3-4,2) -- (0-4,0);

\draw[dashed] (0-4,0) -- (4,0);

\draw[dashed] (0-4,5) -- (4,5);
\draw[dashed] (3-4,7) -- (7,7);

\draw[dashed] (3-4,2)--(0,2);
\draw[dashed] (3,2) -- (7,2);

\draw (1.5,6) to [out=250,in=120] (2.5,4);
\draw (2.5,4) to [out=300,in=40] (1.5,1);

\node (E4) at (2.5+0.2,4.6-0.2) {\textcolor{blue}{$E_4$}};
\node (E4space) at (2.5,4) {\textcolor{blue}{\footnotesize{$\bullet$}}};

\node (X2) at (1.2,4.6) {$[H_1^{cc}]$};
\draw (X2)--(2.04,4.6);

\node (X2) at (1.6,3) {$[H_2^{cc}]$};
\draw (X2)--(2.7,3);

\node (H1) at (-0.5,4) {$[H_1]$};

\node[purple] (H2) at (1,2) {$[H_1^c]$};

\node[purple] (H3) at (2.5,5.5) {$[H_2^c]$};

\node (H4) at (3.5,4) {$[H_2]$};

\node (xi) at (1.5-0.1,6.2) {$\overbrace{\hspace{19.7em}}^{\xi}$};

\node[rotate=34] (X_1) at (1.54,6.2-5.4) {$\underbrace{\hspace{8.7em}}$};
\node at (1.8,0.4) {$|H|=0$};

\end{tikzpicture}

\caption{A slice of a parameter space $\xi$ of a Hamiltonian $H$. The equivalence classes are denoted by square brackets, with the superscript denoting the number of constraints applied to the Hamiltonian. The pink surface corresponds to a phase boundary separating entirely unconstrained Hamiltonians $H$. Hamiltonians $H^c$ are confined to the pink surface and separated with the phase boundary corresponding to the black curve. Finally, Hamiltonians confined to the black curve $H^{cc}$ are separated by the phase boundary corresponding to the blue dot.}
\label{CartoonPic2}
\end{figure}
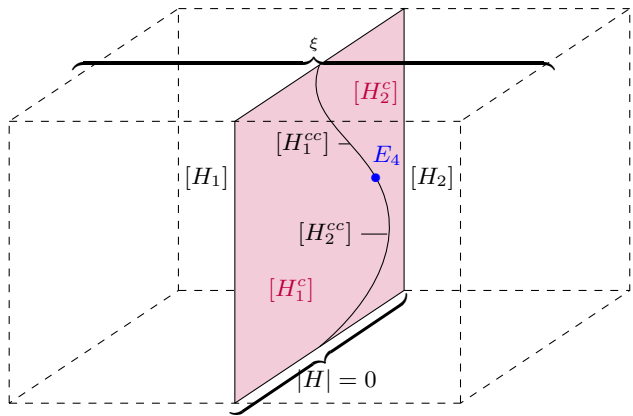

\indent We realised the structure of Fig. \ref{ExpStruc} with a coaxial cable network, which---at radio frequencies---behaves as a tight binding system \cite{cables1,cables2,cables3,classification1,LinearPaper}. Single port measurements with a vector network analyser (VNA) allowed us to map out the local density of states (LDOS), experimentally demonstrating the expected localisation-delocalisation transition. With two-port transmittance measurements, we also observed our predicted transport in our system, which much like observed in Ref. \cite{cables1}, corresponded to a maximum in transmittance at zero energy.

\begin{figure}
\begin{tikzpicture}[scale=0.7]
\node (0) at (0,0) {$\circ$};
\node (1) at (1,1) {$\bullet$};
\node (2) at (1,-1) {$\bullet$};
\node (3) at (2,0) {$\circ$};

\draw (0)--(1)--(3)--(2)--(0);

\node (4) at (4,0) {$\circ$};
\node (5) at (5,1) {$\bullet$};
\node (6) at (5,-1) {{$\bullet$}};
\node (7) at (6,0) {$\circ$};

\draw (4)--(5)--(7)--(6)--(4);

\node (8) at (8,0) {$\circ$};
\node (9) at (9,1) {$\bullet$};
\node (10) at (9,-1) {$\bullet$};
\node (11) at (10,0) {$\circ$};

\draw (8)--(9)--(11)--(10)--(8);

\draw (1)--(4);
\draw (4)--(2);

\draw (7)--(9);
\draw (7)--(10);

\end{tikzpicture}
\caption{This structure was experimentally realised with a coaxial cable network, and was used to corroborate sequential topology. Details of experimental results are included in section \ref{LocsTransmission}.}
\label{ExpStruc}
\end{figure}
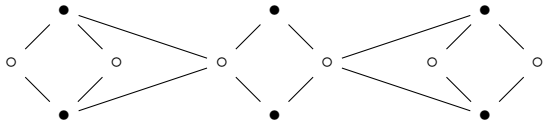

\indent An alternative approach to understanding sequential topology is to start with a system which entirely lacks disorder. 
For example the rotated square lattice of Fig. \ref{RotSquareFig}. For a finite system with no disorder, there are highly (and exactly) degenerate zero energy states (the number of which scales with the number of rows of sites). Allowing individual hopping terms to gain independence from one another reduces this degeneracy, opening up gaps around zero energy. As we increase the number of hopping terms that are independent, we eventually find an entirely gapped system with no zero energy states. Our approach of sequential topology gives a method of understanding precisely (in algebraic terms) how perturbing hopping terms break this degeneracy. Using the number of independent hopping terms as a proxy for disorder strength, sequential topology gives a mathematically precise way to understand how varying amounts of disorder alters topological classification in finite structures.

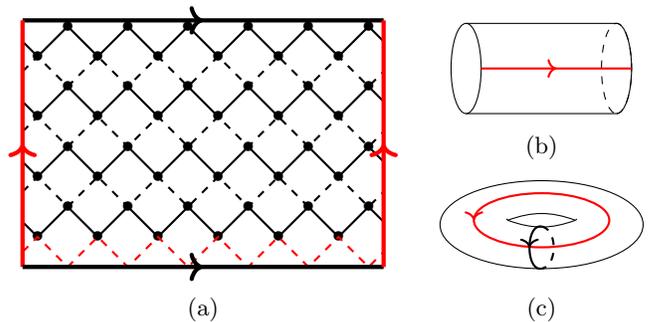
\begin{figure}

\begin{tikzpicture}
\node (ToroidalBC) at (4.5,-1.1) {
\begin{tikzpicture}[scale=0.45]

\draw (0,0) circle [x radius=3, y radius=1*1.3, rotate=0];
\draw[thick,->-,red] (0,0.1*1.3) circle [x radius=2, y radius=0.5*1.6, rotate=0];

\draw (1,0.1*1.3) to [out=160,in=20] (-1,0.1*1.3);

\draw (1,0.1*1.3) to [out=200,in=-20] (-1,0.1*1.3);

\node (a) at (1.3,0.2*1.3) {};
\node (b) at (-1.3,0.2*1.3) {};

\draw (1,0.1*1.3) to [out=20,in=200] (a);
\draw (-1,0.1*1.3) to [out=-20,in=-20] (b);

\draw[thick,->-] (0,-0.11*0.5) to [out=170,in=190] (0,-1*1.3);
\draw[thick,dashed] (0,-0.11*0.5) to [out=-10,in=10] (0,-1*1.3);

\end{tikzpicture}
};

\node (c) at (4.5,-2.2) {(c)};

\node (LoopedBC) at (4.5,1) {
\begin{tikzpicture}[scale=0.4,rotate=270]

\draw[dashed] (0,0) circle [x radius=1.5, y radius=0.5, rotate=0];

\draw (0,-5) circle [x radius=1.5, y radius=0.5, rotate=0];
\draw (-1.5,0) arc
    [
        start angle=180,
        end angle=0,
        x radius=1.5cm,
        y radius =0.5cm
    ] ;

\coordinate (a) at (1.5,0);
\coordinate (aa) at (1.5,-5);

\coordinate (b) at (-1.5,0);
\coordinate (bb) at (-1.5,-5);

\coordinate (c) at (0,+0.5);
\coordinate (cc) at (0,-5+.5);

\draw (a) -- (aa);
\draw (b) -- (bb);
\draw[red,thick,->-] (cc) -- (c);

\end{tikzpicture}
};

\node (b) at (4.5,-0.05) {(b)};

\node (rotsquare) at (0,0) {
\begin{tikzpicture}[scale=0.4]

\coordinate (0Z) at (-0.5,-0.5);
\coordinate (00) at (1,0);
\coordinate (01) at (3,0);
\coordinate (02) at (5,0);
\coordinate (03) at (7,0);
\coordinate (04) at (9,0);
\coordinate (05) at (11,0);
\coordinate (0P) at (11.5,-0.5);

\coordinate (1Z) at (-0.5,-0.5-1);
\coordinate (10) at (0,0-1);
\coordinate (11) at (2,0-1);
\coordinate (12) at (4,0-1);
\coordinate (13) at (6,0-1);
\coordinate (14) at (8,0-1);
\coordinate (15) at (10,0-1);
\coordinate (1P) at (11.5,-0.5-1);

\node (n00) at (1,0) {$\bullet$};
\node (n01) at (3,0) {$\bullet$};
\node (n02) at (5,0) {$\bullet$};
\node (n03) at (7,0) {$\bullet$};
\node (n04) at (9,0) {$\bullet$};
\node (n05) at (11,0) {$\bullet$};

\node (n10) at (0,0-1) {$\bullet$};
\node (n11) at (2,0-1) {$\bullet$};
\node (n12) at (4,0-1) {$\bullet$};
\node (n13) at (6,0-1) {$\bullet$};
\node (n14) at (8,0-1) {$\bullet$};
\node (n15) at (10,0-1) {$\bullet$};

\draw[thick] (0Z) -- (10) -- (00) -- (11) -- (01) -- (12) -- (02) -- (13) -- (03) -- (14) -- (04) -- (15) -- (05) -- (0P);

\coordinate (2Z) at (-0.5,-0.5-2);
\coordinate (20) at (1,0-2);
\coordinate (21) at (3,0-2);
\coordinate (22) at (5,0-2);
\coordinate (23) at (7,0-2);
\coordinate (24) at (9,0-2);
\coordinate (25) at (11,0-2);
\coordinate (2P) at (11.5,-0.5-2);

\coordinate (3Z) at (-0.5,-0.5-3);
\coordinate (30) at (0,0-3);
\coordinate (31) at (2,0-3);
\coordinate (32) at (4,0-3);
\coordinate (33) at (6,0-3);
\coordinate (34) at (8,0-3);
\coordinate (35) at (10,0-3);
\coordinate (3P) at (11.5,-0.5-3);

\node (n00) at (1,0-2) {$\bullet$};
\node (n01) at (3,0-2) {$\bullet$};
\node (n02) at (5,0-2) {$\bullet$};
\node (n03) at (7,0-2) {$\bullet$};
\node (n04) at (9,0-2) {$\bullet$};
\node (n05) at (11,0-2) {$\bullet$};

\node (n10) at (0,0-3) {$\bullet$};
\node (n11) at (2,0-3) {$\bullet$};
\node (n12) at (4,0-3) {$\bullet$};
\node (n13) at (6,0-3) {$\bullet$};
\node (n14) at (8,0-3) {$\bullet$};
\node (n15) at (10,0-3) {$\bullet$};

\draw[thick] (2Z) -- (30) -- (20) -- (31) -- (21) -- (32) -- (22) -- (33) -- (23) -- (34) -- (24) -- (35) -- (25) -- (2P);

\draw[dashed, thick] (1Z) -- (10) -- (20) -- (11) -- (21) -- (12) -- (22) -- (13) -- (23) -- (14) -- (24) -- (15) -- (25) -- (1P);

\coordinate (4Z) at (-0.5,-0.5-4);
\coordinate (40) at (1,0-4);
\coordinate (41) at (3,0-4);
\coordinate (42) at (5,0-4);
\coordinate (43) at (7,0-4);
\coordinate (44) at (9,0-4);
\coordinate (45) at (11,0-4);
\coordinate (4P) at (11.5,-0.5-4);

\coordinate (5Z) at (-0.5,-0.5-5);
\coordinate (50) at (0,0-5);
\coordinate (51) at (2,0-5);
\coordinate (52) at (4,0-5);
\coordinate (53) at (6,0-5);
\coordinate (54) at (8,0-5);
\coordinate (55) at (10,0-5);
\coordinate (5P) at (11.5,-0.5-5);

\node (n00) at (1,0-4) {$\bullet$};
\node (n01) at (3,0-4) {$\bullet$};
\node (n02) at (5,0-4) {$\bullet$};
\node (n03) at (7,0-4) {$\bullet$};
\node (n04) at (9,0-4) {$\bullet$};
\node (n05) at (11,0-4) {$\bullet$};

\node (n10) at (0,0-5) {$\bullet$};
\node (n11) at (2,0-5) {$\bullet$};
\node (n12) at (4,0-5) {$\bullet$};
\node (n13) at (6,0-5) {$\bullet$};
\node (n14) at (8,0-5) {$\bullet$};
\node (n15) at (10,0-5) {$\bullet$};

\draw[thick] (4Z) -- (50) -- (40) -- (51) -- (41) -- (52) -- (42) -- (53) -- (43) -- (54) -- (44) -- (55) -- (45) -- (4P);

\draw[dashed, thick] (3Z) -- (30) -- (40) -- (31) -- (41) -- (32) -- (42) -- (33) -- (43) -- (34) -- (44) -- (35) -- (45) -- (3P);

\coordinate (6Z) at (-0.5,-0.5-6);
\coordinate (60) at (1,0-6);
\coordinate (61) at (3,0-6);
\coordinate (62) at (5,0-6);
\coordinate (63) at (7,0-6);
\coordinate (64) at (9,0-6);
\coordinate (65) at (11,0-6);
\coordinate (6P) at (11.5,-0.5-6);

\coordinate (7Z) at (-0.5,-0.5-7);
\coordinate (70) at (0,0-7);
\coordinate (71) at (2,0-7);
\coordinate (72) at (4,0-7);
\coordinate (73) at (6,0-7);
\coordinate (74) at (8,0-7);
\coordinate (75) at (10,0-7);
\coordinate (7P) at (11.5,-0.5-7);

\node (n00) at (1,0-6) {$\bullet$};
\node (n01) at (3,0-6) {$\bullet$};
\node (n02) at (5,0-6) {$\bullet$};
\node (n03) at (7,0-6) {$\bullet$};
\node (n04) at (9,0-6) {$\bullet$};
\node (n05) at (11,0-6) {$\bullet$};

\node (n10) at (0,0-7) {$\bullet$};
\node (n11) at (2,0-7) {$\bullet$};
\node (n12) at (4,0-7) {$\bullet$};
\node (n13) at (6,0-7) {$\bullet$};
\node (n14) at (8,0-7) {$\bullet$};
\node (n15) at (10,0-7) {$\bullet$};

\draw[thick] (6Z) -- (70) -- (60) -- (71) -- (61) -- (72) -- (62) -- (73) -- (63) -- (74) -- (64) -- (75) -- (65) -- (6P);

\draw[dashed, thick] (5Z) -- (50) -- (60) -- (51) -- (61) -- (52) -- (62) -- (53) -- (63) -- (54) -- (64) -- (55) -- (65) -- (5P);

\coordinate (8Z) at (-0.5,-0.5-8);
\coordinate (80) at (1,0-8);
\coordinate (81) at (3,0-8);
\coordinate (82) at (5,0-8);
\coordinate (83) at (7,0-8);
\coordinate (84) at (9,0-8);
\coordinate (85) at (11,0-8);
\coordinate (8P) at (11.5,-0.5-8);

\draw[red,dashed, thick] (7Z) -- (70) -- (80) -- (71) -- (81) -- (72) -- (82) -- (73) -- (83) -- (74) -- (84) -- (75) -- (85) -- (7P);

\draw[ultra thick,->-] (-0.5,0.2) -- (11.5,0.2);
\draw[red,ultra thick,->-] (-0.5,-8) -- (-0.5,0.2);
\draw[ultra thick,->-] (-0.5,-8) -- (11.5,-8);
\draw[red,ultra thick,->-] (11.5,-8) -- (11.5,0.2);

\end{tikzpicture}
};

\node (a) at (0,-2.2) {(a)};

\end{tikzpicture}
\caption{A rotated square lattice (a), which, for all unit hopping terms and under cylindrical (b) or toroidal (c) boundary conditions the number of zero energy states is exactly twice the number of rows of sites connected with solid hopping terms. As hopping terms are perturbed, the degeneracy of zero energy states is reduced, until under strong enough disorder no zero energy states remain. Sequential topology provides an analytical way to understand exactly how degeneracy in zero energy modes is reduced with increasing hopping disorder. Note that under cylindrical boundary conditions (b) the two red edges are joined, and the red hopping terms are removed. Under toroidal boundary conditions (c) the red edges are joined and the black edges are joined with the red hopping terms. Note that the number of rows or columns may be increased.}
\label{RotSquareFig}
\end{figure}

\indent Beyond the coaxial cable networks discussed herein, there are a number of systems which are sufficiently controllable to experimentally explore sequential topology \cite{OptomechanicalCombDevice,polariton1,polariton2,acoustics1,acoustics2,Topoelectic0,Topoelectric1,Topoelectric2}. 
Or with the alternative view of sequential topology, it may be explored in systems where a specific lattice (such as the previously mentioned rotated square lattice) is put under small amounts of strain \cite{Strain} to perturb specific hopping terms.

\indent For some intuition into sequential topology, it should be noted that there are a number of similarities between sequential topology and higher order topology. 
This is because we are interested in the topological properties of a system where we iteratively increase the number of constraints. This results in studying topological properties of lower dimensional subspaces of the initial parameter space $\xi$, which can be understood in terms of higher topological invariants.

\indent For a $d$ dimensional $n$th order higher order topological insulator, one is interested in $d-n$ dimensional boundary states 
which requires studying higher topological invariants of momentum space  \cite{HOTI1,HOTI2,HOTI3,HOTI4,HOTI5,HOTI6}. 
This is largely a formal similarity, as sequential topology studies lower dimensional features of the parameter space $\xi$, and higher order topology studies lower dimensional features of momentum space \cite{HOTI_LowerDimPropertisOfMomentumSpace}. Indeed higher order topology results in lower dimensional boundary modes in a non-trivial phase. Conversely a phase boundary in sequential topology corresponds to a reduced localisation in zero energy modes.

\indent We begin by giving a brief overview of our approach to studying topology in finite structures before explaining sequential topology and how to calculate the sequential classification in section \ref{FiniteClassificationBackground}. 
We then discuss the physical consequences associated to phase transitions in sequential topology and give an experimental and theoretical case study in section \ref{LocsTransmission}. Finally in section \ref{concl} we conclude and give a brief outlook for further applications of sequential topology. 
A supplementary material is included for further experimental and mathematical details.

\section{Premise of a finite classification}\label{FiniteClassificationBackground}

\noindent Our theory of sequential topology builds directly off our previous approach to finite classification, and so we first explain some of the relevant features of the work explored in Ref. \cite{classification1}.

\indent To study the topology of an arbitrary finite structure, we define a tight binding Hamiltonian $H$ on a graph $G$ so that hopping terms of $H$ are non-zero only on edges of $G$, an example of which is illustrated in Fig. \ref{3SectionExample}. To allow disorder to be introduced to the structure, we first assume each hopping term is an algebraically independent variable taking value in some field (generally $\mathbb{C}$ or $\mathbb{R}$). In order to prevent discontinuous hopping evolution between two Hamiltonians $H_1$ and $H_2$ on $G$ we require that non-zero valued hopping terms cannot be set to zero (the structure cannot be cut) and no zero valued entries of $H$ can become non-zero (no new connections or sites may be introduced). \\
\indent To understand a topological phase transition in this setting, we use unavoidable non-adiabatic evolution to define topological phase boundaries. That is, if, for any continuous path evolving between two Hamiltonians $H(t=0) = H_0$ and $H(t=1) = H_1$ on $G$, there is an unavoidable gap closure in the energy spectrum at some point $H(0<t<1)$ we can be certain $H_0$ and $H_1$ are topologically distinct structures. \\
\indent We can interpret phase boundaries separating topologically distinct Hamiltonians in the parameter space $\xi$. Recall $\xi$ is given by the $\zeta$ dimensional tuple $(u,v,s,\cdots)$ of algebraically independent hopping terms. In $\xi$ phase boundaries correspond to unbounded $\zeta-1$ dimensional surfaces, illustrated by an example of a slice of a parameter space in Fig. \ref{FiniteZerothStepTopology}. \\
\indent Such a finite phase transition is the simplest to understand in chiral structures. A chiral Hamiltonian $H$ is defined by the existence of a unitary $C$ such that $CHC^{\dagger} = -H$. This operation leaves the eigenvalues of $H$ unchanged, ensuring a symmetry between its positive and negative eigenvalues: if $-\varepsilon$ is an eigenvalue of $H$ than so is $\varepsilon$. Evolving a non-singular Hamiltonian to be singular therefore introduces a pair of zero energy states, so that if the path from $H_0$ to $H_1$ unavoidably sets $|H|=0$ then the system has necessarily undergone a topological phase transition. \\ 
\indent In chiral structures, we may therefore use the determinant to study the number of different ways a Hamiltonian may pass through a topological phase boundary. This allows the classification problem to become algebraic: each hopping term is a variable, so the determinant $|H|$ is a polynomial. The number of phases of is then related to the irreducible factorisation of $|H|$ \cite{classification1}. \\
\indent In the case of a structure with all algebraically independent hopping terms, chiral symmetry implies a bipartite structure, where sites may be split in to two partite sets (for example, the black and white sets illustrated in Fig. \ref{ExpStruc} and in Fig. \ref{3SectionExample}). This ensures the Hamiltonian may be permuted to a form
\begin{equation}
H = \begin{pmatrix}
0 & Q \\
Q^{\dagger} & 0
\end{pmatrix}
\end{equation}
where $|H|=-|Q||Q^{\dagger}|$. As each hopping term is algebraically independent, any hopping terms that appears in $|Q|$ is not repeated. That is, the determinant $|Q|$ is linear in each hopping term in its expansion, and has a factorisation $|Q|=\prod_i^M |q_i|$ with $M$ irreducible factors $|q_i|$. 

\indent As a consequence of the factorisation of $|Q|$ such chiral structures have an $N\mathbb{Z}_2$ classification. This is because each irreducible factor has an independent topological invariant. A phase transition occurs by continuously evolving hopping terms through a map that takes $\text{Sign}[|Q|]\lessgtr 0$ to $\text{Sign}[|Q|]\gtrless 0$ which corresponds to one of the factors $|q_i|$ of $|Q|$ undergoing the map $\text{Sign}[|q_i|]\lessgtr 0$ to $\text{Sign}[|q_i|]\gtrless 0$.
Each factor that may be set to zero (which we call a \textit{non-trivial factor}: some factors might not be possible to set to zero) independently defines two topological phases. $N$ non-trivial factors therefore result in an $N\mathbb{Z}_2$ classification with a total of $2^N$ distinct topological phases. Note that $N\leq M$ as some factors may be trivial. \\
\indent We have previously shown that each irreducible factors $|q_i|$ of $|Q|=\prod_{i=1}^M |q_i|$ can be associated to a \textit{section} $g_i$ of the Hamiltonian, which is given by a certain substructure $g$ of $G$ \cite{classification1}. Each section corresponds to a distinct subset of sites, and all the hopping terms that connect the sites of that specific section. An example with three sections labelled $g_1, g2$ and $g_3$ is illustrated in Fig. \ref{3SectionExample}. In the same work we also rigorously proved that the determinant $|Q|$ has this factorisation if and only if the sites can be ordered in such a way that $Q$ has an upper triangular block form. That is,
\begin{equation} \label{TriangularBlockFormOfQ}
Q = \begin{pmatrix}
q_1 & \cdots & c_{1,M-1} & c_{1,M} \\
& \ddots & \vdots & \vdots \\
& & q_{M-1} & c_{M-1,M} \\
\text{\huge0} & & & q_M \\
\end{pmatrix}
\;
\end{equation}
where each entry is a matrix block. This form is of significant importance for calculating the sequential classification of a structure. In this form, each section $g_i$ then corresponds to the structure with a `sub-Hamiltonian'
\begin{equation}
h_i = \begin{pmatrix}
0 & q_i \\
q_i^{\dagger} & 0
\end{pmatrix}.
\end{equation}

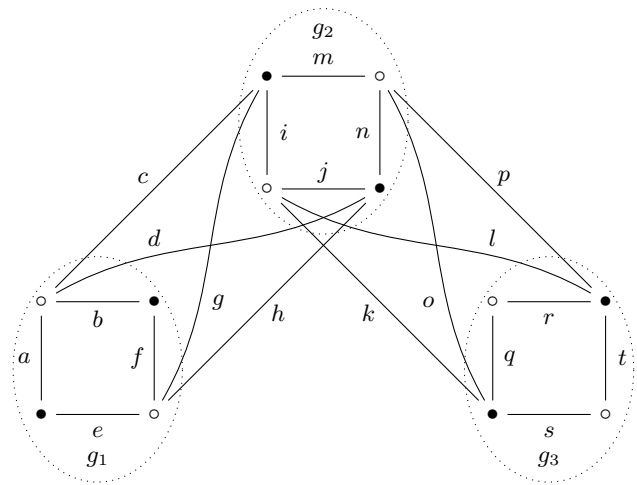
\begin{figure}
\begin{tikzpicture}[scale=1.5]

\node (a) at (0,0) {$\bullet$};
\node (b) at (1,0) {$\circ$};
\node (c) at (1,1) {$\bullet$};
\node (d) at (0,1) {$\circ$};

\draw (a)--(b)--(c)--(d)--(a);

\node (e) at (2,2) {$\circ$};
\node (f) at (3,2) {$\bullet$};
\node (g) at (3,3) {$\circ$};
\node (h) at (2,3) {$\bullet$};

\draw (e)--(f)--(g)--(h)--(e);

\node (i) at (4,0) {$\bullet$};
\node (j) at (5,0) {$\circ$};
\node (k) at (5,1) {$\bullet$};
\node (l) at (4,1) {$\circ$};

\draw (i)--(j)--(k)--(l)--(i);

\draw (b)--(f);
\draw (d)--(h);

\draw (b) to [out=0+45+22.5-7.5,in=180+45+22.5-7.5] (h);
\draw (d) to [out=0+45-22.5+7.5,in=180+45-22.5+7.5] (f);

\draw (e)--(i);
\draw (g)--(k);

\draw (e) to [out=0-45+22.5-7.5,in=180-45+22.5-7.5] (k);
\draw (g) to [out=0-45-22.5+7.5,in=180-45-22.5+7.5] (i);

\node (A) at (-0.15,0.5) {$a$};
\node (B) at (0.5,1-0.15) {$b$};
\node (F) at (1-0.15,0.5) {$f$};
\node (E) at (0.5,-0.15) {$e$};

\node (Q) at (4+0.15,0.5) {$q$};
\node (R) at (4+0.5,1-0.15) {$r$};
\node (T) at (4+1+0.15,0.5) {$t$};
\node (S) at (4+0.5,-0.15) {$s$};

\node (I) at (2+0.15,2+0.5) {$i$};
\node (M) at (2+0.5,2+1+0.15) {$m$};
\node (N) at (2+1-0.15,2+0.5) {$n$};
\node (J) at (2+0.5,2+0.15) {$j$};

\node (C) at (1-0.1,2+0.1) {$c$};
\node (D) at (1,1+0.5+0.07) {$d$};
\node (G) at (1+0.5+0.07,1) {$g$};
\node (H) at (2+0.1,1-0.1) {$h$};

\node (P) at (4+0.1,2+0.1) {$p$};
\node (L) at (4,1+0.5+0.07) {$l$};
\node (O) at (3+0.5-0.07,1) {$o$};
\node (K) at (3-0.1,1-0.1) {$k$};

\draw[dotted] (0.5,0.4) circle [x radius=0.75, y radius=1, rotate=0];
\node (q1) at (0.5,-0.4) {$g_1$};

\draw[dotted] (4+0.5,0.4) circle [x radius=0.75, y radius=1, rotate=0];
\node (q3) at (4+0.5,-0.4) {$g_3$};

\draw[dotted] (2+0.5,3-0.4) circle [x radius=0.75, y radius=1, rotate=0];
\node (q2) at (2+0.5,3+0.4) {$g_2$};
\end{tikzpicture}
\caption{A structure with three sections, denoted and labelled by the dotted ellipses. Each section is non-trivial giving a zeroth step classification of $3\mathbb{Z}_2$. The sequential classification of this structure is given in the diagram \eqref{ExampleSequentialClassification}. Each independent hopping term is labelled, giving a tight binding model where $g_i$ has a corresponding chiral block $q_i$ which is explicitly defined in equation \eqref{Example5}.}
\label{3SectionExample}
\end{figure}

\indent We now wish to generalise this notion of topological phases, where phase boundaries (or further constrained subspaces of $\xi$) may themselves be topologically non-trivial. To understand such phases, consider a simplified picture of a finite topological phase transition, illustrated in Fig. \ref{CartoonPic2}. 
Any Hamiltonian being evolved along a path from on side of the pink surface in Fig. \ref{CartoonPic2} to the other side of the pink surface takes us through a topological phase transition, as this surface corresponds to a topological phase boundary with two zero energy states (that is, $|H|=0$). Now consider if we constrain our Hamiltonian to the pink surface itself, by implementing some algebraic constraint on the hopping terms. There may be a line on this surface (corresponding to four zero energy states) that splits it in to two topologically distinct regions. This ensures the surface corresponding to $|H|=0$ is also topologically non-trivial. \\
\indent We may repeat this process, constraining our Hamiltonian to the line of four zero energy states. If there is a point on this line (for instance, the blue point in Fig. \ref{CartoonPic2}, representing Hamiltonians with six zero energy states) then the line itself is topologically non-trivial. This gives a sequence of topological classifications, where increasing the number of constraints on the Hamiltonian results in a new classification of the structure. We may represent this sequence with the diagram
\begin{equation}
\mathbb{Z}_2 \rightarrow \mathbb{Z}_2 \rightarrow \mathbb{Z}_2
\end{equation}
where each arrow denotes the addition of a single constraint on the hopping terms of the Hamiltonian. \\
\indent It is important to note that sometimes multiple constraints are required to find a non-trivial classification. For example, if the pink surface of Fig. \ref{CartoonPic2} had a point of four zero energy states instead of a line. A Hamiltonian confined to an open line (on the pink surface) that through this point is still topologically non-trivial. We denote such a sequence as
\begin{equation} \label{ExampleSequence2TrivandNonTriv}
\mathbb{Z}_2 \rightarrow 0 \rightarrow \mathbb{Z}_2.
\end{equation}
This may seem like quite an artificial way to implement a topological sequence (why should we choose to constrain the Hamiltonian to a line that intercepts this point?), but is well motivated. As we will show, the sequences we study in this paper are given by the unconstrained topological classification of sub-structures of the original Hamiltonian, so are not arbitrarily chosen.

\indent We use the terminology \textit{number of steps} to denote the number of constrained hopping terms. That is, the $n$th step topology of a Hamiltonian $H$ is the topological classification of a Hamiltonian with $n$ constrained hopping terms. With this terminology a Hamiltonian with all independent hopping terms is classified with the \textit{zeroth step topology}. Although we focus on chiral symmetry in this work, our methods can be readily generalised to systems which lack chiral symmetry.

\indent In order to study sequential topology we generalise our approach to zeroth step topology to the full secular equation
\begin{equation}
|H-\varepsilon I| = \sum_n a_{2n}\lambda^{2n}
\end{equation}
where the restriction to even coefficients is a consequence of the structures' chiral symmetry. For a chiral Hamiltonian the zeroth step topology is found by studying $a_0 = |H|$. A higher steps' topology then corresponds to studying constraints on the hopping terms that solve $a_{2n}=0$ for all $1\leq n\leq m$ for some integer $m$.

\indent We satisfy $a_{2n}=0$ for all $n\leq m$ by imposing a sequence of constraints on the Hamiltonians hopping terms. These take the form of requiring that individual hopping terms satisfy a composition of rational functions we call \textit{constraint maps}. An example of a constraint map is one of the irreducible factors of the Hamiltonian. Forcing hopping terms to satisfy such a factor constrains the Hamiltonian to lie on the pink surface corresponding to $|H|=0$ in Fig. \ref{CartoonPic2}. Importantly all constraint maps are chosen to be an algebraically irreducible function, which is done to ensure all constrained subspaces are simple surfaces. That is, for a particular constrained subspace, then every point on that constrained subspace has the same dimension.

\indent Constraint maps are always analagous to solving $|H|=0$ in zeroth step topology, and indeed are directly related to zeroth step topology of sub-structures of $H$. To see this 
consider an $m$ site tight binding Hamiltonian $H$ with the secular equation $|H-\varepsilon I|$. 
Following the Laplace expansion of a determinant, the coefficient of the $n$th term of the secular equation $a_n$ can be found by summing over the non-zero determinants of the diagonal $(m-n)\times (m-n)$ submatrices of $H$ (that is, the $(m-n)$th principal minors of $H-\varepsilon I$). These are exactly the non-zero determinants of any substructure of $H$ with $n$ deleted sites. 
That is a higher steps' topology is precisely related to the zeroth step topology of substructures of $H$.

\subsection{Calculating a sequential classification} \label{CalcSeqClass}

\noindent We now present a general approach for calculating the sequential classification of a structure. In order to explain our approach, we work through an example to calculate the sequential classification of the structure in Fig. \ref{3SectionExample}. In particular, we will show that this structure has the sequential classification
\begin{equation}
3\mathbb{Z}_2 \rightarrow 0 \rightarrow \mathbb{Z}_2 \rightarrow 0 \rightarrow 0 \rightarrow \mathbb{Z}_2 \rightarrow 0.
\end{equation}
To do this, we will first derive all the relevant constraint maps, and from this infer a function that can calculate an arbitrary constraint map. We will then define a convention for the order in which we apply constraint maps. By using this convention, we find the minimum number of constraint maps necessary to find the topologically interesting subspaces of $\xi$. Full details of a general protocol to calculate the sequential classification of a structure are given in the Appendix section \ref{ConstriantMapProtocol}.

\indent The structure in Fig. \ref{3SectionExample} has three sections, and if each hopping term is restricted to the domain $\mathbb{R}^+$ has the zeroth step classification $3\mathbb{Z}_2$. The tight binding Hamiltonian for this structure is given by
\begin{equation} \label{Example5} \begin{split}
H = &\begin{pmatrix}
0 & Q \\
Q^{\dagger} & 0
\end{pmatrix}, \\
Q = \begin{pmatrix}
q_1 & C_{1,2} & 0 \\
0 & q_2 & C_{2,3} \\
0 & 0 & q_3
\end{pmatrix}& = \begin{pmatrix}
a & b & c & d & 0 & 0 \\
e & f & g & h & 0 & 0 \\
0 & 0 & m & n & p & o \\
0 & 0 & i & j & l & k \\
0 & 0 & 0 & 0 & r & q \\
0 & 0 & 0 & 0 & t & s \\
\end{pmatrix}.
\end{split}
\end{equation}
To begin, the first constraint map is defined by satisfying $a_0=|H|=0$. We may do this by choosing the constraint map $|q_1|=0$ which corresponds to the surface $af-be=0$, and gives two zero energy states. \\
\indent To get four zero energy states in this structure, two further additional constraint maps must be satisfied. The first is defined by setting another section to be singular. This is because $g_1$ can only host two zero energy states, so we get the constraint $|q_2|=0$ or $|q_3|=0$. For the purposes of this example, we take $|q_2|=0$. \\
\indent The second constraint arises from requiring that the zero energy states from both sections are well defined on the entire structure. We call such constraint maps \textit{inter-section constraint maps}. For example, from the section $g_1$ then $Q$ has the zero energy state,
\begin{equation}
\begin{pmatrix}
q_1 & C_{1,2} & C_{1,3} \\
0 & q_2 & C_{2,3} \\
0 & 0 & q_3
\end{pmatrix} \begin{pmatrix}
\psi_1^1 \\
0 \\ 
0
\end{pmatrix} = 0
\end{equation}
and from the section $g_2$
\begin{equation}\label{InterSectionContraints}
\begin{pmatrix}
q_1 & C_{1,2} & C_{1,3} \\
0 & q_2 & C_{2,3} \\
0 & 0 & q_3
\end{pmatrix} \begin{pmatrix}
\psi_2^1 \\
\psi_2^2 \\ 
0
\end{pmatrix} = 0.
\end{equation}
We denote by the subscript $j$ of $\psi^i_j$ the section the state originates from, and the superscript $i$ a remaining section the state is non-zero on. 
For the latter state to be well defined, the equation
\begin{equation} \label{3rdConstraintExample}
q_1\psi_2^1 + C_{1,2}\psi^2_2 = 0
\end{equation}
must be satisfied. In our example this gives rise to a second additional constraint map, defined by
\begin{equation} \label{InterSectionContraints2}
dfm-bhm+bgn-cfn = 0.
\end{equation}
To see the origin of this constraint map, notice that equation \eqref{3rdConstraintExample} has solutions when the vector $-C_{1,2}\psi^2_2$ is a linear combination of the column vectors of $q_1$. This is satisfied precisely by equation \eqref{InterSectionContraints2}. The coefficients of this linear combination then define the terms in $\psi_2^1$. \\
\indent Such a linear combination is possible when $-C_{1,2}\psi^2_2$ is in the column space of $q_1$. This means we can derive the relevant constraint map by replacing (any) column of $q_1$ with $-C_{1,2}\psi^2_2$ and finding the condition for when the resultant matrix is singular. That is, by labelling each column of $q_1$ with $q_1^k$ (where the superscript $k$ denotes the index of the column),
\begin{equation} \label{ConstraintMapDerivation} \begin{split}
\det \left[ q_1 - q_1^k  + \delta \right] = 0, \,\,\,\,\,\,\quad & \\
\delta = \begin{pmatrix}
0 & \cdots & 0 & C_{1,2}\psi^2_2 & 0 & \cdots & 0
\end{pmatrix} &
\end{split} \end{equation}
where $\delta$ is a matrix formed of column vectors, and is non-zero \textit{only} in the $k$th column. This is because the constraint map $|q_1|=0$ ensures that every column of $q_1$ is a linear combination of all the remaining columns \cite{classification1}. Deleting any column thus leads to a matrix with the same column space as $q_1$, and so $-C_{1,2}\psi^2_2$ is in this column space when $\det [q_1 - q_1^k  + \delta ]=0$. All inter-section constraint maps may be derived in an analagous way, with full details of their formula given using the protocol presented in the Appendix section \ref{ConstriantMapProtocol}. \\ 
\indent Satisfying these three total constraint maps gives a structure with four zero energy states. If all but one of these three constraint maps is satisfied, the structure has two zero energy states, and is confined to a surface in $\xi$ that is separated in to two distinct topological phases. The phase boundary corresponds to the final (initially unsatisfied) constraint. This gives the sequential classification of
\begin{equation} \label{First3Steps}
3\mathbb{Z}_2 \rightarrow 0 \rightarrow \mathbb{Z}_2.
\end{equation}
As a visiaul aid, this sequential classification is analogous to the parameter space in Fig. \ref{CartoonPic2} where the first constraint map forces the Hamiltonian on to the pink surface, the second on to the line, which then has two topological phases separated by the blue point. 

\indent We may go further and consider the constraints needed to get six zero energy states, using $g_3$ to generate the necessary constraint maps. Our first constraint is from $|q_3|=0$. We then need the zero energy states from $g_1$, $g_2$ and $g_3$ to be well defined on each section $g_1,g_2,g_3$. For this we get two more constraint maps,
\begin{equation}
\begin{split}
lns-jps+jot-knt&=0 \\
df-bh&=0
\end{split}
\end{equation}
for a total of six constraint maps for six zero energy states. As above, satisfying all but one of these constraints gives a Hamiltonian confined to a surface in $\xi$ which has two topologically distinct phases. In this structure it is not possible to introduce more than six zero energy states, and so the surface in $\xi$ corresponding to a Hamiltonian with six zero energy states is topologically trivial. This gives the complete sequential classification of this structure as
\begin{equation} \label{ExampleSequentialClassification}
3\mathbb{Z}_2 \rightarrow 0 \rightarrow \mathbb{Z}_2 \rightarrow 0 \rightarrow 0 \rightarrow \mathbb{Z}_2 \rightarrow 0.
\end{equation}
\indent As mentioned above, it is important to define an order in which constraint maps are applied. This is because it is possible to apply constraint maps in any order. This can alter the sequence of classifications, for instance by satisfying each inter section constraint map first, the resulting sequence would be
\begin{equation}
0 \rightarrow 0 \rightarrow 0 \rightarrow 3\mathbb{Z}_2 \rightarrow \mathbb{Z}_2 \rightarrow \mathbb{Z}_2 \rightarrow 0.
\end{equation}
To avoid ambiguity in the sequential classification, we use a convention defined as follows:
\begin{enumerate}
\item Start with a chiral structure confined to having $2n$ zero energy states.
\item Choose a non-singular section $g_n$.
\item Maintaining all previous zero energy states in the structure, find the minimum number of constraint maps necessary to provide two additional zero energy states originating from the section $g_n$.
\end{enumerate}
Under this convention, we find (for a specific section $g_n$) the largest dimension subspace of $2n+2$ zero energy states within the $2n$ zero energy state subspace. The codimension of this subspace is then read off simply as the number of constraint maps applied to the Hamiltonian. Applying this convention we get the sequential classification of diagram \eqref{ExampleSequentialClassification}, so that the two zero energy state subspace has codimension one, four zero energy subspace has codimension three, and six zero energy subspace has codimension six. \\
\indent It is no coincidence that (under our convention) the sequential classification is non-trivial when one less than a triangular number of constraint maps are satisfied, and in fact follows directly from the triangular block form of $Q$. To be unambiguous, the origin of the first constraint is from only requiring that $|q_1|=0$. The second and third constraints arise from requiring that $|q_2|=0$ and that the zero energy states from $g_2$ are well defined on $g_1$. The final three constraint maps are from $|q_3|=0$, and requiring the the zero energy states from $g_3$ are well defined on $g_1$ and $g_2$. 
This is the origin of the main theoretical results of this paper: For most chiral structures with $n-1$ singular sections and $2n-2$ zero energy states, there are precisely $n$ additional constraint maps that must be satisfied to introduce $2n$ zero energy states. The total number of constraints to give a chiral structure $2n$ zero energy states is therefore $\sum_{p=1}^n p$: the $n$th triangular number. A proof of this correspondence is given in the Appendix sections \ref{TriangularNumber} and \ref{ConstriantMapProtocol}. We represent this result with the sequence
\begin{widetext}
\begin{equation}
\begin{tikzcd}
N\mathbb{Z}_2 \arrow{r} & 0 \arrow{r} & M\mathbb{Z}_2 \arrow{r} & 0 \arrow{r} & 0 \arrow{r} & P\mathbb{Z}_2 \arrow{r} & 0 \arrow{r} & 0 \arrow{r} & 0 \arrow{r} & R\mathbb{Z}_2 \arrow{r} & 0 \arrow{r} & \cdots
\end{tikzcd}
\end{equation}
This sequence is readily extended to complex chiral Hamiltonians where there are two constraint maps for each step of the real case (one for the real part and one for the complex part of each hopping term). This gives the sequence
\begin{equation}
\begin{tikzcd}
0 \arrow{r} & N\mathbb{Z}_2 \arrow{r} & 0 \arrow{r} & 0 \arrow{r} & M\mathbb{Z}_2 \arrow{r} & 0 \arrow{r} & 0 \arrow{r} & 0 \arrow{r} & 0 \arrow{r} & P\mathbb{Z}_2 \arrow{r} & 0 \arrow{r} & \cdots
\end{tikzcd}
\end{equation}
\end{widetext}
where $N,M,P,R$ are integers. As each constraint map is related to the zeroth step topology of a Hamiltonians substructure, these integers are calculated in an analogous way to the zeroth step classification. Full details of a general protocol to calculate these integers are presented in the Appendix section \ref{ConstriantMapProtocol}. These sequences then terminate when a maximum degeneracy in zero modes is reached.

\indent Finally it is important to note that although for most chiral structures there is a correspondence between a triangular number of constraints and $2n$ zero energy states, there are exactly two cases where this does not hold. To see the origin of such cases, notice that another solution to the constraint map in equation \eqref{ConstraintMapDerivation} is when $\delta =0$. This can occur two ways. The first is when the underlying connectivity of the Hamiltonian ensures that if $q_i \psi_i^i = 0$ then $\psi_i^j$ is automatically zero and similarly if $q_j \psi_j^j = 0$ then $\psi_j^i$ is automatically zero. The second case is when a prior constraint map also sets $\delta=0$. Both cases are also discussed in the Appendix \ref{TriangularNumber}, with the prior case also discussed in greater detail for the experimental case study in section \ref{LocsTransmission}.

\section{The physics of sequential topology: localisation and transport}\label{LocsTransmission}

\noindent In this section, we discuss how iteratively increasing the number of zero energy states alters the localisation and transport properties of a structure. We begin by discussing delocalisation phenomenon while working through the classification of the structure in Fig. \ref{ExpCaseStudy}, before discussing how this affects transport. We then present our experimental results corroborating our sequential classification.

\begin{rem}
For the discussion in this section, it is useful to be able to attribute a particular zero energy state as originating from a specific section. In the case that a section $|q_i|=0$ then we say that the zero energy state on $H$ that satisfies $q_i\psi_i^i=0$ \textit{originates} from the section $g_i$.
\end{rem}

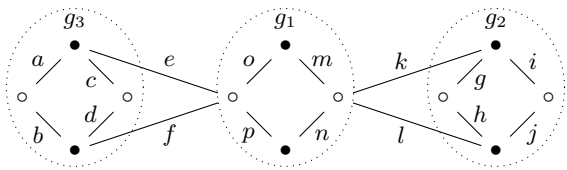
\begin{figure}
\begin{tikzpicture}[scale=0.7]
\node (0) at (0,0) {$\circ$};
\node (1) at (1,1) {$\bullet$};
\node (2) at (1,-1) {$\bullet$};
\node (3) at (2,0) {$\circ$};

\node (a) at (0.5-0.2,0.5+0.2) {$a$};
\node (c) at (1.5-0.2,0.5-0.2) {$c$};
\node (d) at (1.5-0.2,-0.5+0.2) {$d$};
\node (b) at (0.5-0.2,-0.5-0.2) {$b$};

\draw (0)--(1)--(3)--(2)--(0);

\node (4) at (4,0) {$\circ$};
\node (5) at (5,1) {$\bullet$};
\node (6) at (5,-1) {{$\bullet$}};
\node (7) at (6,0) {$\circ$};

\node (o) at (0.5-0.2+4,0.5+0.2) {$o$};

\node (m) at (1.5+0.2+4,0.5+0.2) {$m$};
\node (n) at (1.5+0.2+4,-0.5-0.2) {$n$};

\node (p) at (0.5-0.2+4,-0.5-0.2) {$p$};

\draw (4)--(5)--(7)--(6)--(4);

\node (8) at (8,0) {$\circ$};
\node (9) at (9,1) {$\bullet$};
\node (10) at (9,-1) {$\bullet$};
\node (11) at (10,0) {$\circ$};

\node (g) at (0.5+0.2+8,0.5-0.2) {$g$};
\node (h) at (1.5+0.2+8,0.5+0.2) {$i$};
\node (i) at (1.5+0.2+8,-0.5-0.2) {$j$};
\node (j) at (0.5+0.2+8,-0.5+0.2) {$h$};

\draw (8)--(9)--(11)--(10)--(8);

\draw (1)--(4);
\draw (4)--(2);

\draw (7)--(9);
\draw (7)--(10);

\node (e) at (3-0.2,0.5+0.2) {$e$};
\node (f) at (3-0.2,-0.5-0.2) {$f$};

\node (k) at (4+3+0.2,0.5+0.2) {$k$};
\node (l) at (4+3+0.2,-0.5-0.2) {$l$};

\draw[dotted] (1,0+0.2) circle [x radius=1.3, y radius=1.5, rotate=0];
\node (q1) at (1,1.45) {$g_3$};

\draw[dotted] (5,0+0.2) circle [x radius=1.3, y radius=1.5, rotate=0];
\node (q1) at (5,1.45) {$g_1$};

\draw[dotted] (9,0+0.2) circle [x radius=1.3, y radius=1.5, rotate=0];
\node (q1) at (9,1.45) {$g_2$};

\end{tikzpicture}
\caption{The structure of our theoretical and experimental case study. Each independent hopping term is labelled with an algebraic variable. Sections are labelled as $g_1$, $g_2$, and $g_3$, and encircled within each ellipse. The chiral matrix block $q_i$ corresponding to each section $g_i$ is shown explicitly in equation \eqref{HamiltonianForExperiment}.}
\label{ExpCaseStudy}
\end{figure}

\indent The structure in Fig. \ref{ExpCaseStudy} has the tight binding Hamiltonian
\begin{equation} \label{HamiltonianForExperiment} \begin{split}
H = &\begin{pmatrix}
0 & Q \\
Q^{\dagger} & 0
\end{pmatrix}, \\
Q = \begin{pmatrix}
q_1 & C_{1,2} & C_{1,3} \\
0 & q_2 & 0 \\
0 & 0 & q_3
\end{pmatrix}& = \begin{pmatrix}
m & n & k & l & 0 & 0 \\
o & p & 0 & 0 & e & f \\
0 & 0 & g & h & 0 & 0 \\
0 & 0 & i & j & 0 & 0 \\
0 & 0 & 0 & 0 & a & b \\
0 & 0 & 0 & 0 & c & d \\
\end{pmatrix}.
\end{split}
\end{equation}
For all positive real hopping terms this structure has three non-trivial sections corresponding to a $3\mathbb{Z}_2$ zeroth step classification. Two of these sections are entirely disconnected, so any zero energy states originating from $g_3$ automatically have no support on $g_2$ and simultaneously any zero energy states originating from $g_2$ automatically have no support on $g_3$. This is one of the two previously mentioned cases where fewer than a triangular number of constraints are necessary to increase the degeneracy of zero energy modes. Under our convention this gives two sequential classifications, of either
\begin{equation} \label{Seq1}
\begin{tikzcd}
3\mathbb{Z}_2 \arrow{r} & \mathbb{Z}_2 \arrow{r} & 0  \arrow{r} & 0 \arrow{r} & \mathbb{Z}_2 \arrow{r} & 0
\end{tikzcd}
\end{equation}
or with a different choice in the order of setting sections singular
\begin{equation} \label{Seq2}
\begin{tikzcd}
3\mathbb{Z}_2 \arrow{r} & 0 \arrow{r} & \mathbb{Z}_2 \arrow{r} & 0 \arrow{r} & \mathbb{Z}_2 \arrow{r} & 0.
\end{tikzcd}
\end{equation}
Both sequences satisfy our convention for the order of applying constraint maps. \\
\indent The difference between the two sequences arises because if $|q_1|=0$ then the corresponding zero energy state has no support on $q_2$ (and vice versa), reducing the number of necessary constraint maps. 
Setting $g_1$ and $g_2$ to being singular before $g_3$ gives, by our convention, the first sequence, and any other choice gives the second sequence. It is possible to map between the two sequences without altering the number of zero energy states by continuously evolving hopping terms (this changes which constraint maps are satisfied). 

\begin{rem}
As we can choose an order to apply constraint maps, any individual constraint map may be chosen to increase the degeneracy of zero energy states. It is therefore helpful to carefully define a sequential topological phase boundary. We consider a Hamiltonian to be on a topological phase boundary if enough constraint maps are satisfied for $2n$ zero energy states, but no additional constraint maps are satisfied.
\end{rem}

\indent On a topological phase boundary, there is an observable delocalisation of the zero energy states. In particular, if a section provides two zero energy states to a Hamiltonian, then zero energy states have 
non-zero support on every site of that section. This is because every constraint map is related to the zeroth step topology of a substructure of the Hamiltonian, which we have previously shown results in such a delocalisation \cite{classification1}. Experimental results confirming this generalisation to sequential topology are displayed in Fig. \ref{CoreGraphExpResults}, with the sections $g_1$ and $g_2$ being singular. \\
\indent The triangular form of $Q$ also ensures the zero energy states have an unusual spread throughout the rest of the structure. For convenience, we recall this form here. Consider a generic Hamiltonian $H$ with a chiral block
\begin{equation}
Q = \begin{pmatrix}
q_1 & \cdots & c_{1,j} & \cdots & c_{1,M} \\
& \ddots & \vdots & \vdots & \vdots \\
& & q_{j} & \cdots & c_{j,M} \\
& & & \ddots & \vdots \\
\text{\huge0} & & & & q_M \\
\end{pmatrix}
\;
\end{equation}
containing the hopping terms from black to white sites. If $|q_i|=0$ then the zero energy state originating on the black sites of $g_i$ has the form
\begin{equation} \label{AssumpSupport2}
Q\begin{pmatrix}
\psi_j^1 \\
\vdots \\
\psi_j^{j-1} \\
\psi_j^j \\
0 \\
\vdots \\
0
\end{pmatrix} = 0
\end{equation}
so is non-zero only on $g_i$ and sections $g_j$ for $i<j$. Conversely on the block $Q^{\dagger}$ containing hopping terms from the white to the black sites,
\begin{equation} \label{AssumpSupport2}
Q^{\dagger}\begin{pmatrix}
0 \\
\vdots \\
0 \\
\phi_j^j \\
\phi_j^{j+1} \\
\vdots \\
\phi_j^{M} \\
\end{pmatrix} = 0
\end{equation}
so that the state originating on the white sites of $g_j$ is non-zero only on the white sites of $g_j$ and sections $g_k$ such that $k>j$. \\
\indent Due to the triangular block form of $Q$ it is convenient to define a partial ordering of sections, where $g_i=g_j$ if $q_i$ and $q_j$ may be permuted to swap places on the diagonal of $Q$ without disrupting the triangular block form, otherwise $g_i<g_j$ if $i<j$. We say two sections are \textit{equal} if $g_i=g_j$ and that for $g_i<g_j$ that $g_j$ is \textit{upper} of $g_i$ and $g_i$ is \textit{lower} of $g_j$. \\
\indent The split in the compact localisation of zero energy states means that a non-singular section only hosts zero energy states on one sublattice, unless there is simultaneously a section above and a section below it that each provide two zero energy states to the structure. In the case of the structure in Fig. \ref{ExpCaseStudy} there are three sections, with the partial ordering $g_1<g_2=g_3$. Therefore the only way support of zero energy states may be found on both sublattice of a section in this structure, is if that section provides two zero energy states to the Hamiltonian. This gives an experimental route to probe the sequential topology of this structure: if the LDOS at zero energy is non-zero on both sublattices of a section, we know that section provides two zero energy states to the structure, and thus we can count the degeneracy of zero modes.

\subsection{Experimental corroboration of sequential topology}

\noindent We realise the structure in Fig. \ref{ExpCaseStudy} using a coaxial cable network, which if all cables are of the same transmittance time $\tau$ maps on to a tight binding model \cite{LinearPaper,classification1} (at radio frequencies). Sites correspond to junctions in the network, and cables to hopping terms. The value of a hopping term $H_{A,B}$ between two sites $A,B$ is related to the reciprocal of the impedance of the cable connecting the two sites. This allows highly controllable tight binding Hamiltonians to be experimentally realised. The Hamiltonian has solutions of $H\psi = \varepsilon\psi$ for an `energy' $\varepsilon=\cos\omega\tau$ with $\omega$ the driving frequency, and entries in the eigenstate $\psi$ scaled voltages at junctions in the network. \\
\indent Measurements on the coaxial cable network were performed with a two-port vector network analyser (a NanoVNA V2 Plus 4 Pro). States of the Hamiltonian can be mapped with a single port measurement of reflectance, which is used to find the local impedance and is proportional to the local density of states (LDOS). Two port measurements of transmittance allow us to find the ratio of a state on two sites, and can be used as an experimental signature of a structure being topologically marginal \cite{LinearPaper,classification1}.

The structures in Fig. \ref{LDOSexpsresults} follow the sequential classification in the diagram \eqref{Seq2}. Each structure satisfying the additional constraint map for each step in the sequence. The order of applying constraint maps was to set the hopping terms to allow $g_1$, then $g_2$, followed by $g_3$ to each provide zero energy states to the structure, with the exact constraint maps detailed in the Appendix section \ref{ExpDets}. The Appendix section \ref{ExpDets} also containts a similar set of experimental results that corroborate the sequential classification in the diagram \eqref{Seq1}.

All maxima in the LDOS where observed within an energy window of $\Delta\varepsilon = 0.096$ (corresponding to a frequency window of $7.02$MHz
) with maxima having a standard deviation of 0.02 (or 1.57 MHz) indicating our tight binding description models the coaxial cable system very well. LDOS were calculated directly from raw data, being integrated over the range $\varepsilon\in [-0.102,0.102]$. Data agreed well with the predicted compact localisation, as displayed in Fig. \ref{LDOSexpsresults}. Measurements on sites with zero support corresponded to an absence of a local maxima in the LDOS, with non-zero values found only due to the experimental spectra never truly going to zero. The significant non-zero value for one of the sites in the upper left panel of Fig. \ref{LDOSexpsresults} (a) was a consequence of the broadening of the LDOS spectra, where low energy states had significant support on this site.

\begin{figure}
\begin{tikzpicture}[scale=0.7]

\node (0) at (0-5,0-6) {};
\node (1) at (1-5,1-6) {};
\node (2) at (1-5,-1-6) {};
\node (3) at (2-5,0-6) {};


\draw (0) -- (1);
\draw[dashed] (1)--(3);
\draw (3)--(2);
\draw[dashed] (2)--(0);

\node (4) at (4-5,0-6) {};
\node (5) at (5-5,1-6) {};
\node (6) at (5-5,-1-6) {};
\node (7) at (6-5,0-6) {};


\draw (6)--(4)--(5);
\draw[dashed] (5)--(7)--(6);

\node (8) at (8-5,0-6) {};
\node (9) at (9-5,1-6) {};
\node (10) at (9-5,-1-6) {};
\node (11) at (10-5,0-6) {};


\draw (8)--(10)--(11);
\draw[dashed] (11)--(9)--(8);

\draw (1)--(4);
\draw[dashed] (4)--(2);

\draw[dashed] (7)--(9); 
\draw (7)--(10);

\draw[fill=white] (0-5,0-6) circle [x radius=0.5*0.453, y radius=0.5*0.453, rotate=0];
\draw[fill=black] (1-5,1-6) circle [x radius=0.5*0.0289, y radius=0.5*0.0289, rotate=0];
\draw[fill=black] (1-5,-1-6) circle [x radius=0.5*0.0332, y radius=0.5*0.0332, rotate=0];
\draw[fill=white] (2-5,0-6) circle [x radius=0.5*0.00486, y radius=0.5*0.00486, rotate=0];

\draw[fill=white] (0-5+4,0-6) circle [x radius=0.5*0.464, y radius=0.5*0.464, rotate=0];
\draw[fill=black] (1-5+4,1-6) circle [x radius=0.5*0.854, y radius=0.5*0.854, rotate=0];
\draw[fill=black] (1-5+4,-1-6) circle [x radius=0.5*0.854, y radius=0.5*0.854, rotate=0];
\draw[fill=white] (2-5+4,0-6) circle [x radius=0.5*0.161, y radius=0.5*0.161, rotate=0];

\draw[fill=white] (0-5+4+4,0-6) circle [x radius=0.5*0.916, y radius=0.5*0.916, rotate=0];
\draw[fill=black] (1-5+4+4,1-6) circle [x radius=0.5*0.334, y radius=0.5*0.334, rotate=0];
\draw[fill=black] (1-5+4+4,-1-6) circle [x radius=0.5, y radius=0.5, rotate=0];
\draw[fill=white] (2-5+4+4,0-6) circle [x radius=0.5*0.925, y radius=0.5*0.925, rotate=0];

\node (a) at (0,-8) {(a)};

\node (alpha) at  (2-5+0.3,0-6) {$\alpha$};
\node (beta) at  (0-5+4,0-6+0.45) {$\beta$};

\node at (0,-12) {
\includegraphics[scale=0.4]{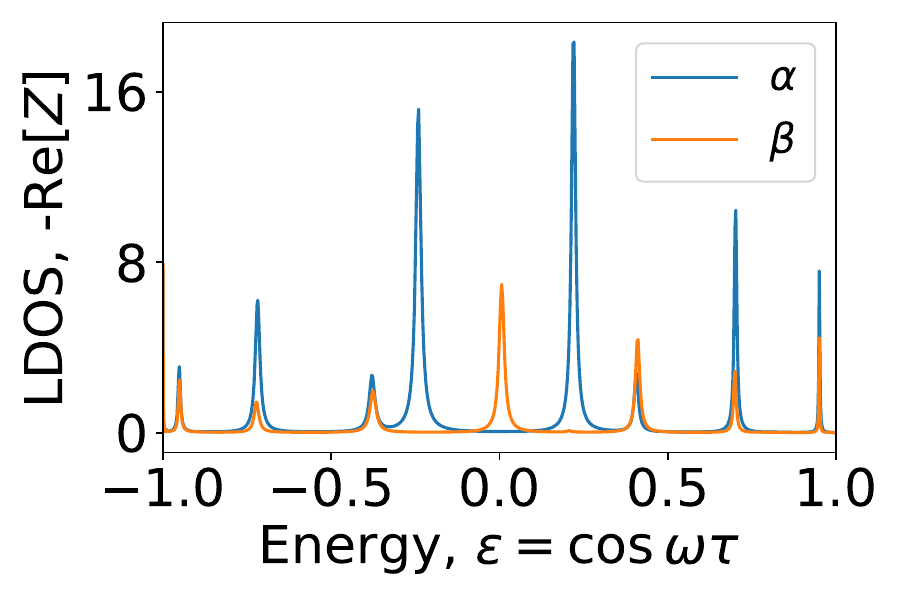}
};

\node (b) at (0,-15.5) {(b)};

\end{tikzpicture}

\caption{(a) Displays the experimentally measured LDOS on every site of a structure on a sequential topological phase boundary, with four zero energy states. LDOS were calculated directly from experimental data at $\varepsilon=0$, with the experimental spectra for the sites $\alpha$ and $\beta$ displayed in (b). The node diameter is proportional to the LDOS integrated over an energy range of $\varepsilon\in [-0.102,0.102]$ corresponding to a frequency range of 111MHz to 126MHz. Observe how the zero energy states have non-zero support on every site of $g_1$ and $g_2$, but only on one site of $g_3$, corroborating the prediction that on a sequential topological phase boundary all sites of singular sections support zero energy states. 50$\Omega$ RG58 cables are indicated with the dashed lines, and $93\Omega$ RG62 cables with the solid lines. The hopping terms were randomly selected from a binary distribution constrained so that the structure is on a sequential topological phase boundary, with four zero energy states.}
\label{CoreGraphExpResults}
\end{figure}

\begin{widetext}

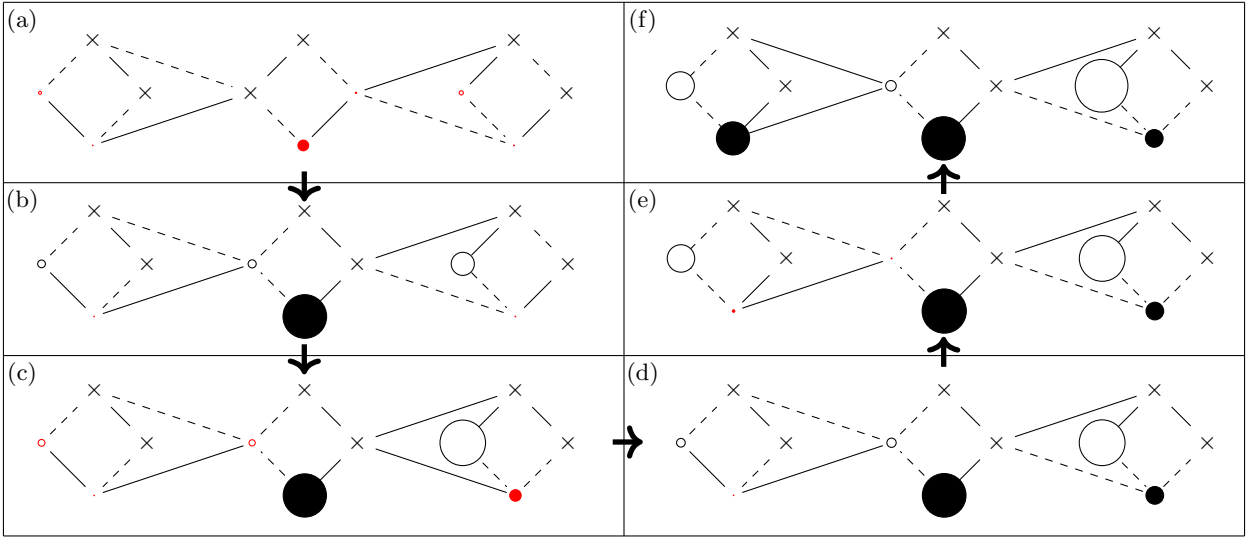
\begin{figure}

\begin{tikzpicture}

\draw (5.75,1.25+0.6) -- (5.75,-5.25);
\draw (-3+0.5,1.25+0.6) -- (-3+0.5,-5.25);
\draw (14,1.25+0.6) -- (14,-5.25);

\draw (-2.5,-5.25) -- (14,-5.25);
\draw (-2.5,1.25+0.6) -- (14,1.25+0.6);
\draw (-2.5,-1.7-1.15) -- (14,-1.7-1.15);

\draw (-2.5,-1.7-1.15+2.3) -- (14,-1.7-1.15+2.3);

\node (a) at (-3+0.75,1.25+0.35) {(a)};
\node (a) at (-3+0.75,-1.7-1.15+2.3-0.25) {(b)};
\node (a) at (-3+0.75,-1.7-1.15-0.25) {(c)};

\node (a) at (6,1.25+0.35) {(f)};
\node (a) at (6,-1.7-1.15+2.3-0.25) {(e)};
\node (a) at (6,-1.7-1.15-0.25) {(d)};


\draw [-To,line width=2pt] (5.75-0.15,-4) -> (5.75+0.25,-4);


\node (0) at (1+0.5+0.05,0+0.6+0.1) {

\begin{tikzpicture}[scale=0.7]

\node (0) at (0-5,0-6) {};
\node (1) at (1-5,1-6) {$\times$};
\node (2) at (1-5,-1-6) {};
\node (3) at (2-5,0-6) {$\times$};


\draw[dashed] (0) -- (1);
\draw (1)--(3);
\draw[dashed] (3)--(2);
\draw (2)--(0);

\node (4) at (4-5,0-6) {$\times$};
\node (5) at (5-5,1-6) {$\times$};
\node (6) at (5-5,-1-6) {};
\node (7) at (6-5,0-6) {};


\draw[dashed] (6)--(4);
\draw (4)--(5);
\draw[dashed] (5)--(7);
\draw (7)--(6);

\node (8) at (8-5,0-6) {};
\node (9) at (9-5,1-6) {$\times$};
\node (10) at (9-5,-1-6) {};
\node (11) at (10-5,0-6) {$\times$};


\draw[dashed] (8)--(10);
\draw (10)--(11);
\draw[dashed] (11)--(9);
\draw (9)--(8);

\draw[dashed] (1)--(4);
\draw (4)--(2);

\draw (7)--(9); 
\draw[dashed] (7)--(10);

\draw[color = red, fill=white] (0-5,0-6) circle [x radius=0.0287, y radius=0.0287, rotate=0]; 

\draw[color = red, fill=red] (1-5,-1-6) circle [x radius=0.0056, y radius=0.0056, rotate=0]; 

\draw[color=red,fill=red] (1-5+4,-1-6) circle [x radius=0.1021, y radius=0.1021, rotate=0]; 
\draw[color=red,fill=white] (2-5+4,0-6) circle [x radius=0.0126, y radius=0.0126, rotate=0]; 

\draw[color=red,fill=white] (0-5+4+4,0-6) circle [x radius=0.04, y radius=0.04, rotate=0]; 
\draw[color=red,fill=red] (1-5+4+4,-1-6) circle [x radius=0.007, y radius=0.007, rotate=0]; 

\end{tikzpicture}

};


\draw [-To,line width=2pt] (1.5,-1.7-1.15+2.3+0.15) -> (1.5,-1.7-1.15+2.3-0.25);

\node (1) at (1+0.5+0.07,-2+0.3+0.03) {

\begin{tikzpicture}[scale=0.7]

\node (0) at (0-5,0-6) {};
\node (1) at (1-5,1-6) {$\times$};
\node (2) at (1-5,-1-6) {};
\node (3) at (2-5,0-6) {$\times$};

\draw[dashed] (0) -- (1);
\draw (1)--(3);
\draw[dashed] (3)--(2);
\draw (2)--(0);

\node (4) at (4-5,0-6) {};
\node (5) at (5-5,1-6) {$\times$};
\node (6) at (5-5,-1-6) {};
\node (7) at (6-5,0-6) {$\times$};

\draw[dashed] (6)--(4)--(5);
\draw (5)--(7)--(6);

\node (8) at (8-5,0-6) {};
\node (9) at (9-5,1-6) {$\times$};
\node (10) at (9-5,-1-6) {};
\node (11) at (10-5,0-6) {$\times$};

\draw[dashed] (8)--(10);
\draw (10)--(11);
\draw[dashed] (11)--(9);
\draw (9)--(8);

\draw[dashed] (1)--(4);
\draw (4)--(2);

\draw (7)--(9); 
\draw[dashed] (7)--(10);

\draw[color=red,fill=black] (1-5,-1-6) circle [x radius=0.0053, y radius=0.0053, rotate=0]; 
\draw[fill=white] (0-5,0-6) circle [x radius=0.0736, y radius=0.0736, rotate=0]; 

\draw[fill=black] (1-5+4,-1-6) circle [x radius= 0.4145, y radius= 0.4145, rotate=0]; 
\draw[fill=white] (0-5+4,0-6) circle [x radius=0.077 , y radius=0.077 , rotate=0]; 

\draw[fill=white] (0-5+4+4,0-6) circle [x radius=0.2196, y radius=0.2196, rotate=0]; 

\draw[color=red,fill=black] (1-5+4+4,-1-6) circle [x radius=0.0064, y radius=0.0064, rotate=0]; 

\end{tikzpicture}

};

\draw [-To,line width=2pt] (1.5,-1.7-1.15+0.15) -> (1.5,-1.7-1.15-0.25);

\draw [-To,line width=2pt]  (10,-1.7-1.15-0.15) -> (10,-1.7-1.15+0.25);

\draw [-To,line width=2pt]  (10,-1.7-1.15+2.3-0.15) -> (10,-1.7-1.15+2.3+0.25);

\node (2) at (1+0.5+0.07,-4-0.05) {

\begin{tikzpicture}[scale=0.7]

\node (0) at (0-5,0-6) {};
\node (1) at (1-5,1-6) {$\times$};
\node (2) at (1-5,-1-6) {};
\node (3) at (2-5,0-6) {$\times$};


\draw[dashed] (0) -- (1);
\draw (1)--(3);
\draw[dashed] (3)--(2);
\draw (2)--(0);

\node (4) at (4-5,0-6) {};
\node (5) at (5-5,1-6) {$\times$};
\node (6) at (5-5,-1-6) {};
\node (7) at (6-5,0-6) {$\times$};


\draw[dashed] (6)--(4)--(5);
\draw (5)--(7)--(6);

\node (8) at (8-5,0-6) {};
\node (9) at (9-5,1-6) {$\times$};
\node (10) at (9-5,-1-6) {};
\node (11) at (10-5,0-6) {$\times$};


\draw[dashed] (8)--(10)--(11);
\draw (11)--(9)--(8);

\draw[dashed] (1)--(4);
\draw (4)--(2);

\draw (7)--(9); 
\draw (7)--(10);

\draw[color=red,fill=white] (0-5,0-6) circle [x radius=0.0648, y radius=0.0648, rotate=0]; 
\draw[color=red,fill=black] (1-5,-1-6) circle [x radius=0.0055, y radius=0.0055, rotate=0]; 

\draw[color=red,fill=white] (0-5+4,0-6) circle [x radius=0.0592, y radius=0.0592, rotate=0]; 
\draw[fill=black] (1-5+4,-1-6) circle [x radius=0.4109, y radius=0.4109, rotate=0]; 

\draw[fill=white] (0-5+4+4,0-6) circle [x radius=0.4325, y radius=0.4325, rotate=0]; 
\draw[color=red,fill=red] (1-5+4+4,-1-6) circle [x radius= 0.1094, y radius= 0.1094, rotate=0]; 

\end{tikzpicture}

};

\node (3) at (10+0.07,-4-0.05) {

\begin{tikzpicture}[scale=0.7]

\node (0) at (0-5,0-6) {};
\node (1) at (1-5,1-6) {$\times$};
\node (2) at (1-5,-1-6) {};
\node (3) at (2-5,0-6) {$\times$};


\draw[dashed] (0) -- (1);
\draw (1)--(3);
\draw[dashed] (3)--(2);
\draw (2)--(0);

\node (4) at (4-5,0-6) {};
\node (5) at (5-5,1-6) {$\times$};
\node (6) at (5-5,-1-6) {};
\node (7) at (6-5,0-6) {$\times$};


\draw[dashed] (6)--(4)--(5);
\draw (5)--(7)--(6);

\node (8) at (8-5,0-6) {};
\node (9) at (9-5,1-6) {$\times$};
\node (10) at (9-5,-1-6) {};
\node (11) at (10-5,0-6) {$\times$};


\draw[dashed] (8)--(10)--(11);
\draw (11)--(9)--(8);

\draw[dashed] (1)--(4);
\draw (4)--(2);

\draw (7)--(9); 
\draw[dashed] (7)--(10);

\draw[fill=white] (0-5,0-6) circle [x radius=0.0826, y radius=0.0826, rotate=0]; 

\draw[color=red,fill=red] (1-5,-1-6) circle [x radius=0.0052, y radius=0.0052, rotate=0]; 

\draw[fill=white] (0-5+4,0-6) circle [x radius=0.0883, y radius=0.0883, rotate=0]; 

\draw[fill=black] (1-5+4,-1-6) circle [x radius=0.4147, y radius=0.4147, rotate=0]; 

\draw[fill=white] (0-5+4+4,0-6) circle [x radius=0.4349, y radius=0.4349, rotate=0]; 

\draw[fill=black] (1-5+4+4,-1-6) circle [x radius=0.1683, y radius=0.1683, rotate=0]; 

\end{tikzpicture}
};

\node (4) at (10+0.03,-2+0.3+0.1) {\begin{tikzpicture}[scale=0.7]

\node (0) at (0-5,0-6) {};
\node (1) at (1-5,1-6) {$\times$};
\node (2) at (1-5,-1-6) {};
\node (3) at (2-5,0-6) {$\times$};


\draw[dashed] (0) -- (1);
\draw (1)--(3);
\draw (3)--(2);
\draw[dashed] (2)--(0);

\node (4) at (4-5,0-6) {};
\node (5) at (5-5,1-6) {$\times$};
\node (6) at (5-5,-1-6) {};
\node (7) at (6-5,0-6) {$\times$};


\draw[dashed] (6)--(4)--(5);
\draw (5)--(7)--(6);

\node (8) at (8-5,0-6) {};
\node (9) at (9-5,1-6) {$\times$};
\node (10) at (9-5,-1-6) {};
\node (11) at (10-5,0-6) {$\times$};


\draw[dashed] (8)--(10)--(11);
\draw (11)--(9)--(8);

\draw[dashed] (1)--(4);
\draw (4)--(2);

\draw (7)--(9); 
\draw[dashed] (7)--(10);

\draw[fill=white] (0-5,0-6) circle [x radius=0.2608, y radius=0.2608, rotate=0]; 

\draw[color=red,fill=red] (1-5,-1-6) circle [x radius=0.0239, y radius=0.0239, rotate=0]; 

\draw[color=red,fill=white] (0-5+4,0-6) circle [x radius=0.008, y radius=0.008, rotate=0]; 

\draw[fill=black] (1-5+4,-1-6) circle [x radius=0.4202, y radius=0.4202, rotate=0]; 

\draw[fill=white] (0-5+4+4,0-6) circle [x radius=0.4342, y radius=0.4342, rotate=0]; 

\draw[fill=black] (1-5+4+4,-1-6) circle [x radius=0.1681, y radius=0.1681, rotate=0]; 

\end{tikzpicture}

};

\node (5) at (10+0.02,0+0.6+0.1) {

\begin{tikzpicture}[scale=0.7]

\node (0) at (0-5,0-6) {};
\node (1) at (1-5,1-6) {$\times$};
\node (2) at (1-5,-1-6) {};
\node (3) at (2-5,0-6) {$\times$};

\draw[dashed] (0) -- (1);
\draw (1)--(3);
\draw (3)--(2);
\draw[dashed] (2)--(0);

\node (4) at (4-5,0-6) {};
\node (5) at (5-5,1-6) {$\times$};
\node (6) at (5-5,-1-6) {};
\node (7) at (6-5,0-6) {$\times$};


\draw[dashed] (6)--(4)--(5);
\draw (5)--(7)--(6);

\node (8) at (8-5,0-6) {};
\node (9) at (9-5,1-6) {$\times$};
\node (10) at (9-5,-1-6) {};
\node (11) at (10-5,0-6) {$\times$};


\draw[dashed] (8)--(10)--(11);
\draw (11)--(9)--(8);

\draw (1)--(4);
\draw (4)--(2);

\draw (7)--(9); 
\draw[dashed] (7)--(10);

\draw[fill=white] (0-5,0-6) circle [x radius= 0.2663, y radius= 0.2663, rotate=0]; 

\draw[fill=black] (1-5,-1-6) circle [x radius=0.3148, y radius=0.3148, rotate=0]; 

\draw[fill=white] (0-5+4,0-6) circle [x radius=0.1003, y radius=0.1003, rotate=0]; 

\draw[fill=black] (1-5+4,-1-6) circle [x radius=0.4122, y radius=0.4122, rotate=0]; 

\draw[fill=white] (0-5+4+4,0-6) circle [x radius=0.5, y radius=0.5, rotate=0]; 

\draw[fill=black] (1-5+4+4,-1-6) circle [x radius=0.1683, y radius=0.1683, rotate=0]; 

\end{tikzpicture}
};


\end{tikzpicture}
\caption{The local density of states, calculated as an integral of the experimentally measured impedance spectra over an energy window of $\varepsilon\in [-0.102,0.102]$ with a corresponding frequency range of 111 MHz to 126 MHz, apart from the two red sites in (e) where gapped low lying states were observed in this region and so the frequency range 113MHz to 117MHz was used instead. Otherwise the range was chosen to account for shifts in LDOS maxima, arising from variance in cable length and from T-connectors. 50$\Omega$ RG58 cables are indicated with a dashed line, and 93$\Omega$ RG62 cables with a solid line. Two sites were measured on each section, with an $\times$ depicting unmeasured sites. On measured sites the diameter of the circle is proportional to the calculated integral, rescaled to the maximum measured value out of the entire sequence (the white site on $g_2$ of the final structure). To distinguish a true zero energy state from broadened low lying states, red indicates if the LDOS at $\varepsilon=0$ had a local minima in the experimentally measured LDOS. The arrows denote the position of the structure to the corresponding classification in the sequence \eqref{Seq2}. On a topological phase boundary, a section that has support of a zero energy state on both a black site and a white site indicates two zero energy states originate from that section allowing the number of zero energy states, singular sections, and if the structure is on a higher steps phase boundary to be inferred. Reading off the figure, (b), (d), (f) are on a topological phase boundary, confirming the predicted sequential classification. The non-zero value calculated in (a) is a consequence of broadened peaks of low lying states in the measured spectra, with a local minima found at $\varepsilon=0$ in experiment.}
\label{LDOSexpsresults}
\end{figure}

\end{widetext}

\indent The delocalisation described above also allows us to use a maximum in transmittance as an experimental signature of when a structure is topologically marginal. We may do this by cutting a structure on a hopping term that connects two different sections (for example, the cut displayed in Fig. \ref{HigherStepPhaseBoundaryTransmission} (a)). If the section with the cut site is singular, an zero energy eigenstate persists which has the same voltage on both the input and output sites. For transmittance to be non-zero in a coaxial cable network, both the voltage and current output must be non-zero which corresponds to at least one neighbour of the input and output site having a non-zero voltage. Due to the triangular block structure of $Q$ this occurs on a phase boundary, where both of the sections connected by the cut hopping term provide two zero energy states. 
Due to chiral symmetry this gives a transfer matrix of
\begin{equation}
\begin{pmatrix}
I_{\text{out}} \\
V_{\text{out}}
\end{pmatrix} = M\begin{pmatrix}
I_{\text{in}} \\
V_{\text{in}}
\end{pmatrix} = \begin{pmatrix}
1 & 0 \\
0 & \alpha
\end{pmatrix} \begin{pmatrix}
I_{\text{in}} \\
V_{\text{in}}
\end{pmatrix}
\end{equation}
at zero energy. Assuming there are no losses $|M| = \alpha =1$ hence the current input also matches the current output, resulting in a local maximum in transmittance at $\varepsilon=0$. Away from a phase boundary, a maximum in transmittance at $\varepsilon=0$ is not expected. This is because not every neighbour of the input and output sites will have a non-zero voltage, and thus non-zero transmittance arises only due to broadening of nearby transmittance maximum. As chiral symmetry ensures a symmetric transmittance spectra around $\varepsilon=0$, a local minimum is expected in experiment. These predictions are confirmed in our structure, with experimental results displayed in Fig. \ref{HigherStepPhaseBoundaryTransmission}.

\begin{figure}
\begin{tikzpicture}
\node (0) at (0,0) {\includegraphics[width=0.38\paperwidth]{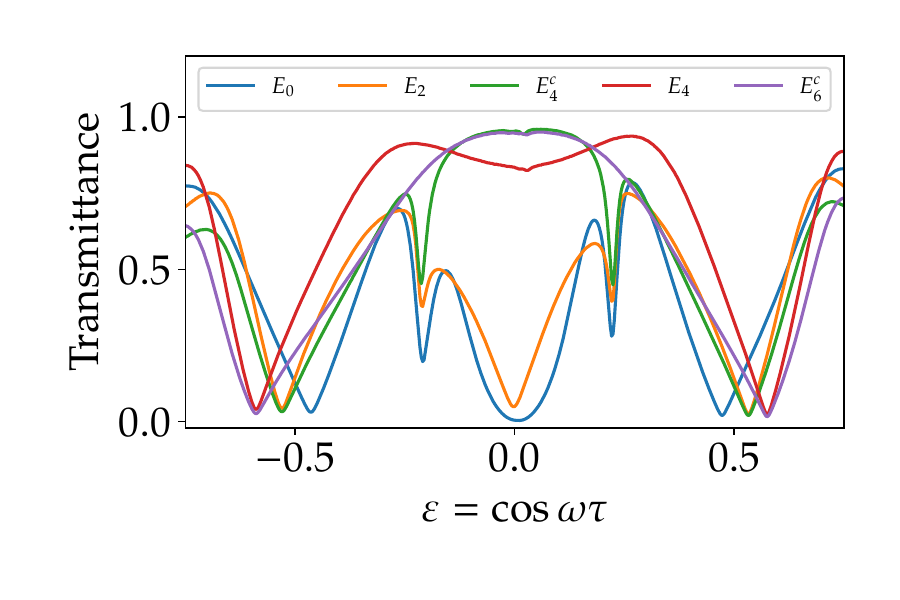}};
\node (A) at (0,-2.7+0.2) {(a)};

\node (1) at (0,-5+0.4) {
\begin{tikzpicture}[scale=0.6]

\node (0) at (0-5,0-6) {$\circ$};
\node (1) at (1-5,1-6) {$\bullet$};
\node (2) at (1-5,-1-6) {$\bullet$};
\node (3) at (2-5,0-6) {$\circ$};

\draw (2)--(0)--(1);
\draw (1)--(3)--(2);

\node (4) at (4-5,0-6) {$\circ$};
\node (5) at (5-5,1-6) {$\bullet$};
\node (6) at (5-5,-1-6) {$\bullet$};
\node (7) at (6-5,0-6) {$\circ$};


\draw (6)--(4)--(5);
\draw (5)--(7)--(6);

\node (8) at (8-5,0-6) {$\circ$};
\node (9) at (9-5,1-6) {$\bullet$};
\node[blue] (10) at (9-5,-1-6) {$\bullet$};
\node (11) at (10-5,0-6) {$\circ$};


\draw (8)--(10)--(11);
\draw (11)--(9)--(8);

\draw (1)--(4);
\draw (4)--(2);

\draw (7)--(9);

\node[blue] (12) at (9-5,-2.5-6) {$\bullet$};

\draw (7)--(12);
\draw[red,dotted] (7)--(10);

\draw[dashed] (9-5,-1.75-6) circle [x radius=0.25, y radius=1.1, rotate=0];
\node (InAndOut) at (9-5+1.3,-1.75-6+0.3) {Input/};
\node (InAndOut) at (9-5+1.3,-1.75-6-0.3) {Output};

\end{tikzpicture}
};

\node (B) at (0,-6.5+0.4) {(b)};

\end{tikzpicture}
\caption{(a) Displays the experimentally measured transmittance spectra for the same structures as in Fig. \ref{LDOSexpsresults}, but with the red hopping term in (b) cut to produce an extra site, which become the two blue sites and are used for an input and output to measure two-site transmittance. $E_n$ denotes the number of exactly zero energy states in the structure, with the superscript $c$ denoting if the structure is on a sequential topological phase boundary. As predicted, a local maxima at $\varepsilon=0$ is observed only for $E_4^c$ and $E_6^c$, both exceeding 94\%. The small dip at $\varepsilon=0$ is a consequence of non uniform cable lengths inducing small amounts of chiral symmetry breaking.}
\label{HigherStepPhaseBoundaryTransmission}
\end{figure}
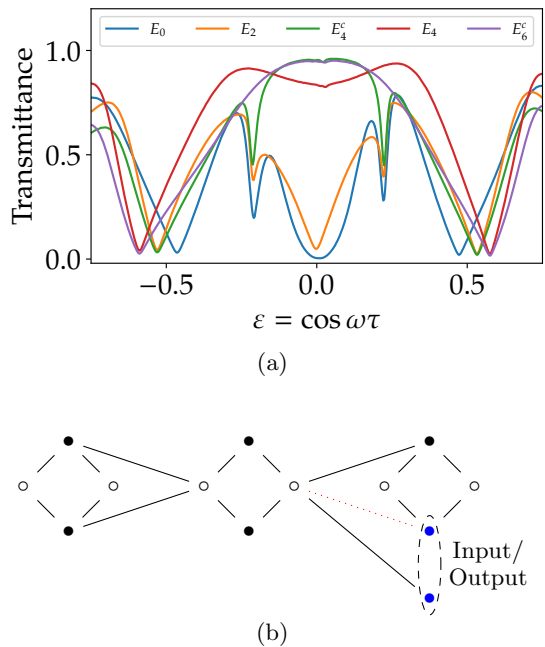

Transmittance data agreed well with prediction, with transmittance at a local maxima for $\varepsilon=0$ only on a topological phase boundary, where it exceeded $94\%$. 
As expected transmittance was a minima at $\varepsilon=0$ when not on a phase boundary, further indicating our iteration through the sequential phase boundaries was indeed what we observed.

In combination, the observed LDOS and transmittance data agree well with our sequential classification of the structure, strongly validating our theory.

\section{Conclusion}\label{concl}

\noindent We have demonstrated that the topological classification of strongly disordered tight binding Hamiltonians is heavily dependant on how one controls the structure, and can be substantially altered by putting constraints on the hopping terms. By defining a convention in which to constrain the hopping terms of a Hamiltonian, we derived a protocol to algorithmically derive the necessary constraints as rational functions of the hopping terms. This protocol can be used as a method to precisely alter the localisation of zero modes and associated scattering response of a particular structure with a minimal number of controlled hopping terms. Using the consequences in localisation and transport, we experimentally confirmed our sequential classification of a twelve site structure, where hopping terms can be independently manipulated.

While the BBC dominates the classification once a structure is of a sufficient size, in thermodynamically small systems strong disorder can overcome the BBC. This is the regime in which sequential topology becomes of significance, and is the scale at which many experimental platforms---beyond coaxial cable networks---are currently realised  \cite{OptomechanicalCombDevice,polariton1,polariton2,acoustics1,acoustics2,Mechanical1,Topoelectic0,Topoelectric1,Topoelectric2}. We believe that the exact control to where zero modes have compact support that sequential topology realises will give much technological application to such systems.

We also note that there are applications of sequential topology to non-Hermitian topology. If we interpret the individual chiral blocks of our Hermitian Hamiltonian as an individual non-Hermitian Hamiltonian \cite{NH1,NH2}, then $n$ singular sections correspond to an $n$th order exceptional point in the parameter space $\xi$, and the phase transitions we study in this paper correspond to the diabolical subspaces. 

\section{Acknowledgements}

\noindent We are enormously grateful for help in the experimental work from Qingqing
Duan, and Ben Kinvig, and also to Belle Darling and
Phillip Graham for huge help in finding space to run the experiments.
Many thanks also to Patrick Fowler, Barry Pickup, Qingqing Duan, Ben
Kinvig, and Elena Callus for many illuminating and insightful
discussions while completing this work. MMM wishes to thank the EPSRC for a PhD studentship (Project Reference 2482664).

\bibliography{SequentialTopology}

\widetext
\pagebreak
\begin{center}
\textbf{\large Supplementary Materials: {Sequential topology: iterative topological phase transitions in finite chiral structures}}
\end{center}
\setcounter{equation}{0}
\setcounter{figure}{0}
\setcounter{table}{0}
\setcounter{page}{1}
\setcounter{section}{0}
\renewcommand\thesection{\Alph{section}}
\renewcommand\thesubsection{\thesection.\arabic{subsection}}
\makeatletter
\renewcommand{\theequation}{S\arabic{equation}}
\renewcommand{\thefigure}{S\arabic{figure}}
\renewcommand{\bibnumfmt}[1]{[S#1]}
\renewcommand{\citenumfont}[1]{S#1}

\section{A triangular number of constraint maps} \label{TriangularNumber}

\noindent In order demonstrate that, for many chiral structures, a triangular number of constraint maps are necessary to increase the degeneracy of zero energy states we first show that, at a minimum, a triangular number of constraint maps are necessary. We then provide a protocol that reaches is minimum in the Appendix section \ref{ConstriantMapProtocol}. To show that, for many chiral structures, a minimum of a triangular number of constraint maps are necessary, we first make an assumption about how zero energy states localise in the chiral block of a Hamiltonian. We call this the \textit{assumption of support} defined in definition \ref{TriangularNumberConstraintsCondition}. We will go beyond this assumption later on. \\
\indent Recall that in a chiral structure with $M$ sections, the Hamiltonian is necessarily possible to permute so that the chiral block has a block triangular form. That is
\begin{equation} \label{TriangularBlockFormOfQ}
Q = \begin{pmatrix}
q_1 & \cdots & c_{1,M-1} & c_{1,M} \\
& \ddots & \vdots & \vdots \\
& & q_{M-1} & c_{M-1,M} \\
\text{\huge0} & & & q_M \\
\end{pmatrix}.
\end{equation}
Recall that the section $g_i$ is \textit{lower} than $g_j$ if $i<j$ and that $g_j$ is \textit{upper} than $g_i$.

\begin{definition} \label{TriangularNumberConstraintsCondition}
For each pair of sections $q_i,q_j$ (with the partial relation $i<j$) then if $|q_j|=0$ and $q_i$ has non-zero support of the nullstate from $q_j$, then we can write the nullstate of the matrix $Q$ as
\begin{equation}
\ket{\Psi^j} = \begin{pmatrix}
A_1\psi^j_1 \\
A_2\psi^j_2 \\
\vdots \\
A_{j-1}\psi^j_{j-1} \\
\psi^j_j \\
0 \\
\vdots \\
0
\end{pmatrix}.
\end{equation}
Where $A_j$ is some non zero matrix, and $\psi^j_j$ is the nullstate of $q_j$. Let this be known as the \textit{assumption of support}.
\end{definition}

\begin{prop} \label{eachsectionindependent}
Consider a finite chiral structure where the nullity of each matrix block $q_i$ is less than or equal to one. Assume that a Hamiltonian has $2n-2$ zero energy states, and that the assumption of support holds. 
Selecting a non-singular section $g_n$ to increase the number of zero energy states to $2n$, then there are at least $n$ further constraints to get $2n$ zero energy states.
\end{prop}

\begin{proof}
Suppose $Q$ has $n-1$ (non-defective) null vectors, resulting $2n-2$ zero energy states of $H$. Index each nullvector by $\ket{\Psi^j}$ and left null vector by $\bra{\Psi^j_L}$ when the section $j$ is the origin of this (left) null vector. For a new null vector from a section $g_n$ with the chiral block $q_n$ and null vector $\ket{\Psi^n}$ to satisfy $Q$ we need $\bra{\Psi_L^j}\ket{\Psi^n} = 0$ for all left null vectors $\bra{\Psi_L^j}$. This ensures the relevant component of $\ket{\Psi^n}$ is in the column space of each singular section, and defines a set of $n-1$ equations.

\indent Under the assumption of support, for each pair of sections $q_i,q_j$ with the partial relation $i<j$ then each $q_i$ has non-zero support of the nullstate from $q_j$ given by $A_{i}\ket{\psi^j_i}$. 
Therefore, every constraint to ensure the orthogonality of the original $n-1$ left null vectors of $Q$ to this additional null vector depend individually on each $\ket{\psi^j}$. Each $\ket{\psi^j}$ can be changed by each section's chiral block $q_j$ so we can alter each constraint independently. That is, the orthogonality condition for each left null vector and the new null vector is independent,
\begin{equation}
\bra{\Psi_L^j}\ket{\Psi^n} = 0 \not\Rightarrow \bra{\Psi_L^i}\ket{\Psi^n} = 0.
\end{equation}
So under the assumption of support there are at least $n-1$ independent equations of this form, and one from $|q_n|=0$ giving at least $n$ algebraically independent constraint maps that need to be satisfied to increase the nullity of $H$. \\
\indent Iterating this bound on the number of constraints, we see there are at least $\sum_1^n m = n(n+1)/2$ total constraints to get a nullity of $2n$.
\end{proof}

\indent Now we have a minimum number of constraints to get a certain nullity in our $N\mathbb{Z}_2$ structure, when the conditions in definition \ref{TriangularNumberConstraintsCondition} is fulfilled, 
follows the triangular numbers. In the following section we will present our protocol that exactly satisfy this lower bound. This protocol therefore provides a constructive proof that for most finite chiral structures, the sequential classification is non-trivial with the application of one less than a triangular number of constraint maps. 

\begin{rem}
As discussed in the main text, when the assumption of support is broken, then it can be possible to require fewer than a triangular number of constraint maps to increase the number of zero energy states. This occurs in one of two cases. The first case is when a prior constraint map prevents a zero energy state originating originating from an upper section to a lower section. The second case is when two or more sections are equal in partial ordering, this automatically ensures zero energy states from each of the equal sections will have no support on one another.
\end{rem}

\section{A protocol for finding constraint maps} \label{ConstriantMapProtocol}

\noindent We now present a protocol that exactly meets the minimum bound for the number of constraint maps discussed in section \ref{TriangularNumberConstraintsCondition}. For this protocol it is important to note that how zero energy states localise as a consequence of the triangular block form of $Q$ in equation \eqref{TriangularBlockFormOfQ}. That is, when a section $|q_j|=0$ is singular, then under the assumption of support, for a zero energy states originating from $q_j$ to be defined over all of $Q$ a zero energy state originating from $g_j$ also has support on $q_i$. Conversely when $|q_i|=0$, then under the assumption of support zero energy states originating from $q_i$ do not have support on $q_j$.

\indent We denote by $c_{j,i}$ the matrix connecting the $j$th section to the $i$th section in $Q$, and $\phi^j_i$ is the support of the null vector originating from the $j$th section on the $i$th section. For this protocol we select a non-singular section $g_n$, to differentiate zero energy states originating from $g_n$ we specifically use $\psi^n_i$ is the support of the nullvector from $g_n$ on the $i$th section. For the chiral block of a section, $q_i$, then $q_i^{l\neq k}$ is $q_i$ with the $k$th column deleted. \\
\indent The protocol for calculating the necessary constraint maps is as follows:
\begin{enumerate} \label{ProtocolForHigherOrderTopology}
\item Require $q_n$ to be non-singular (for now). Let $\bar{q}_n$ be the biadjacency matrix of $g_n$ where $|\bar{q}_n|=0$. 
\item Compute the set of null vectors that satisfy $Q$.
\item Start with the section $g_j$ most below $g_n$.
\begin{enumerate}
\item Compute
\begin{equation} \label{constraint1}
\bar{\phi}^j_n = \sum_{i>n} c_{n,i}\phi^j_i
\end{equation}
where the sum is taken over all the upper sections of $g_n$.
\item Delete a column $\bar{q}_n^k$ of $\bar{q}_n$, replacing it with $\bar{\phi}^j_n$. Compute
\begin{equation}
\det \begin{pmatrix}
\bar{q}_n^{l\neq k} & \bar{\phi}^j_n
\end{pmatrix}.
\end{equation}
\item Choose an irreducible factor of this equation, this defines a constraint map. Satisfy this on the Hamiltonian.
\item Repeat, increasing the sections via the partial ordering, until the constraint maps are satisfied for all the upper singular sections of $g_n$.
\end{enumerate}
\item Set $|q_n|=0$
\end{enumerate}

\noindent Now we turn to the sections lower than $g_n$.

\begin{enumerate}
\item We now require $|q_n| = 0$, so set $q_n=\bar{q_n}$. 
\item Begin the the section $g_j$ that is the least above $g_n$.
\begin{enumerate}
\item Compute
\begin{equation} \label{constraint2}
\bar{\psi}^n_i = \sum_{g_n<g_i<g_j} c_{j,i}\psi^n_i
\end{equation}
where the sum is taken over all the sections lower than $g_n$ but upper of $g_j$.
\item Delete a column $q_j^k$ of $q_j$, replacing it with $\bar{\psi}^n_i$. Compute
\begin{equation}
\det \begin{pmatrix}
\bar{q}_j^{l\neq k} & \bar{\psi}^n_i
\end{pmatrix}
\end{equation}
\item Choose an irreducible factor of this equation, this defines a constraint map. Satisfy this constraint on the Hamiltonian.
\item Repeat, for all singular increasingly lower sections $g_i$. 
\end{enumerate}
\end{enumerate}

\noindent Each constraint map of the form of \eqref{constraint1} exactly satisfies the condition for the relevant component of the zero energy states from all upper sections to be in the column space of $g_n$. Furthermore each constraint map of the form of \eqref{constraint2} exactly satisfies the condition for the relevant component of the zero energy state from $g_n$ to be in the column space of every lower section. This defines $n-1$ constraint maps, so along with the constraint map $|q_n|=0$ exactly satisfies the minimum number of necessary constraint maps under the assumption of support. \\
\indent As discussed above, the assumption of support is broken precisely when a prior constraint map prevents a zero energy state having support on a singular section, or when two (or more) sections are equal by the partial ordering. In both cases this automatically sets the term in equations \eqref{constraint1} and \eqref{constraint2} to zero, and so this protocol accounts for these two cases. Consequently, the protocol automatically finds the optimal (that is, minimum) number of necessary constraint maps to find a new topological classification/increase the degeneracy of zero energy modes.

\noindent This procedure defines a set of constraint maps that increase the degeneracy of zero energy states by two. Even though this procedure defines an order to satisfy constraint maps, once we have the constraint maps we can apply them in any order. If all but one constraint map is satisfied, then the Hamiltonian is confined to a topologically non-trivial subspace of $\xi$. Formally this subspace has a non-trivial zeroth homotopy group if the relevant higher steps' phase boundary is removed. \\
\indent Note that the classification if then given by the number of non-trivial irreducible factors of the final constraint map (whichever constraint map is chosen).

\section{Further experimental details} \label{ExpDets}

\indent As each homotopic sequence falls in to two types (illustrated in the diagrams \eqref{Seq1} and \eqref{Seq2}) we consider two paths through the homotopic structure, one path for each type. The first path uses boundary operators for $g_2$ then $g_1$ then $g_3$. The constraint maps were satisfied in the following order:
\begin{equation}
\begin{split}
|q_2| = gj-ih&=0 \\
|q_1| = ad - bc &= 0 \\
eb-fa&=0 \\
kh-lg &= 0 \\
|q_3| = mp - on &= 0.
\end{split}
\end{equation}
This sequence is experimentally corroborated with the LDOS data in Fig. \ref{LDOSexpsresults}. \\
\indent The second path uses boundary operators for $g_3$ then $g_1$ then $g_2$. Boundary maps were taken in the following order:
\begin{equation}
\begin{split}
|q_3| = mp - on &= 0 \\
|q_2| = gj-ih&=0 \\
eb-fa&=0 \\
|q_1| = ad - bc &= 0 \\
kh-lg&=0
\end{split}
\end{equation}
This sequence is experimentally corroborated with the LDOS data in Fig. \ref{LDOSexpsresults_OtherSequence}.

\indent We also note that both sequences, \eqref{Seq1} and \eqref{Seq2}, are possible to continuously map to one another on a topological phase boundary without increasing the number of zero energy states. This is a general feature of sequential topology, and is a consequence of constraint maps defining connected subspaces. We can incorporate this property into our diagram with vertical arrows. That is, 
\begin{equation} \label{homotopicsequences}
\begin{tikzcd}
3\mathbb{Z}_2 \arrow{d}\arrow{r} & \mathbb{Z}_2 \arrow{d}\arrow{r} & 0 \arrow{dr}\arrow{r} & 0 \arrow{r} & \mathbb{Z}_2 \arrow{r} & 0 \arrow{d} \\
3\mathbb{Z}_2 \arrow{r} & 0 \arrow{r} & \mathbb{Z}_2 \arrow{r} & 0 \arrow{r} & \mathbb{Z}_2 \arrow{r} & 0.
\end{tikzcd}
\end{equation}
This diagram incorporates the entire sequential classification of this structure.

\begin{figure}

\begin{tikzpicture}

\draw (5.75,1.25+0.6) -- (5.75,-5.25);
\draw (-3+0.5,1.25+0.6) -- (-3+0.5,-5.25);
\draw (14,1.25+0.6) -- (14,-5.25);

\draw (-2.5,-5.25) -- (14,-5.25);
\draw (-2.5,1.25+0.6) -- (14,1.25+0.6);
\draw (-2.5,-1.7-1.15) -- (14,-1.7-1.15);

\draw (-2.5,-1.7-1.15+2.3) -- (14,-1.7-1.15+2.3);


\node (a) at (-3+0.75,1.25+0.35) {(a)};
\node (a) at (-3+0.75,-1.7-1.15+2.3-0.25) {(b)};
\node (a) at (-3+0.75,-1.7-1.15-0.25) {(c)};

\node (a) at (6,1.25+0.35) {(f)};
\node (a) at (6,-1.7-1.15+2.3-0.25) {(e)};
\node (a) at (6,-1.7-1.15-0.25) {(d)};

\draw [-To,line width=2pt] (5.75-0.15,-4) -> (5.75+0.25,-4);


\node (0) at (1+0.5+0.05,0+0.6+0.1) {

\begin{tikzpicture}[scale=0.7]



\node (0) at (0-5,0-6) {};
\node (1) at (1-5,1-6) {$\times$};
\node (2) at (1-5,-1-6) {};
\node (3) at (2-5,0-6) {$\times$};


\draw[dashed] (0) -- (1);
\draw (1)--(3);
\draw[dashed] (3)--(2);
\draw (2)--(0);

\node (4) at (4-5,0-6) {$\times$};
\node (5) at (5-5,1-6) {$\times$};
\node (6) at (5-5,-1-6) {};
\node (7) at (6-5,0-6) {};


\draw[dashed] (6)--(4);
\draw (4)--(5);
\draw[dashed] (5)--(7);
\draw (7)--(6);

\node (8) at (8-5,0-6) {};
\node (9) at (9-5,1-6) {$\times$};
\node (10) at (9-5,-1-6) {};
\node (11) at (10-5,0-6) {$\times$};


\draw[dashed] (8)--(10);
\draw (10)--(11);
\draw[dashed] (11)--(9);
\draw (9)--(8);

\draw[dashed] (1)--(4);
\draw (4)--(2);

\draw (7)--(9); 
\draw[dashed] (7)--(10);




\draw[color = red, fill=white] (0-5,0-6) circle [x radius=0.0287, y radius=0.0287, rotate=0]; 
\draw[color = red, fill=red] (1-5,-1-6) circle [x radius=0.0056, y radius=0.0056, rotate=0]; 

\draw[color=red,fill=red] (1-5+4,-1-6) circle [x radius=0.1021, y radius=0.1021, rotate=0]; 
\draw[color=red,fill=white] (2-5+4,0-6) circle [x radius=0.0126, y radius=0.0126, rotate=0]; 

\draw[color=red,fill=white] (0-5+4+4,0-6) circle [x radius=0.04, y radius=0.04, rotate=0]; 
\draw[color=red,fill=red] (1-5+4+4,-1-6) circle [x radius=0.007, y radius=0.007, rotate=0]; 


\end{tikzpicture}

};


\draw [-To,line width=2pt] (1.5,-1.7-1.15+2.3+0.15) -> (1.5,-1.7-1.15+2.3-0.25);

\node (1) at (1+0.5+0.07,-2+0.3+0.03) {

\begin{tikzpicture}[scale=0.7]



\node (0) at (0-5,0-6) {};
\node (1) at (1-5,1-6) {$\times$};
\node (2) at (1-5,-1-6) {};
\node (3) at (2-5,0-6) {$\times$};


\draw[dashed] (0) -- (1);
\draw (1)--(3);
\draw[dashed] (3)--(2);
\draw (2)--(0);

\node (4) at (4-5,0-6) {};
\node (5) at (5-5,1-6) {$\times$};
\node (6) at (5-5,-1-6) {};
\node (7) at (6-5,0-6) {$\times$};


\draw[dashed] (6)--(4);
\draw (4)--(5);
\draw[dashed] (5)--(7);
\draw (7)--(6);

\node (8) at (8-5,0-6) {};
\node (9) at (9-5,1-6) {$\times$};
\node (10) at (9-5,-1-6) {};
\node (11) at (10-5,0-6) {$\times$};


\draw[dashed] (8)--(10)--(11);
\draw (11)--(9)--(8);

\draw[dashed] (1)--(4);
\draw (4)--(2);

\draw[dashed] (7)--(9); 
\draw[dashed] (7)--(10);




\draw[red,fill=white] (0-5,0-6) circle [x radius=0.0305, y radius=0.0305, rotate=0]; 
\draw[fill=red] (1-5,-1-6) circle [x radius=0.0066, y radius=0.0066, rotate=0]; 

\draw[red,fill=white] (0-5+4,0-6) circle [x radius=0.0119, y radius=0.0119, rotate=0]; 
\draw[fill=red] (1-5+4,-1-6) circle [x radius=0.0034, y radius=0.0034, rotate=0]; 

\draw[fill=white] (0-5+4+4,0-6) circle [x radius=0.4217, y radius=0.4217, rotate=0]; 
\draw[fill=black] (1-5+4+4,-1-6) circle [x radius=0.1687, y radius=0.1687, rotate=0]; 


\end{tikzpicture}

};

\draw [-To,line width=2pt] (1.5,-1.7-1.15+0.15) -> (1.5,-1.7-1.15-0.25);

\draw [-To,line width=2pt]  (10,-1.7-1.15-0.15) -> (10,-1.7-1.15+0.25);

\draw [-To,line width=2pt]  (10,-1.7-1.15+2.3-0.15) -> (10,-1.7-1.15+2.3+0.25);

\node (2) at (1+0.5+0.07,-4-0.05) {

\begin{tikzpicture}[scale=0.7]

\node (0) at (0-5,0-6) {};
\node (1) at (1-5,1-6) {};
\node (2) at (1-5,-1-6) {$\times$};
\node (3) at (2-5,0-6) {$\times$};


\draw[dashed] (0) -- (1);
\draw[dashed] (1)--(3);
\draw (3)--(2);
\draw (2)--(0);

\node (4) at (4-5,0-6) {};
\node (5) at (5-5,1-6) {$\times$};
\node (6) at (5-5,-1-6) {};
\node (7) at (6-5,0-6) {$\times$};


\draw[dashed] (6)--(4);
\draw (4)--(5);
\draw[dashed] (5)--(7);
\draw (7)--(6);

\node (8) at (8-5,0-6) {};
\node (9) at (9-5,1-6) {$\times$};
\node (10) at (9-5,-1-6) {};
\node (11) at (10-5,0-6) {$\times$};


\draw[dashed] (8)--(10)--(11);
\draw (11)--(9)--(8);

\draw (1)--(4);
\draw (4)--(2);

\draw[dashed] (7)--(9); 
\draw[dashed] (7)--(10);




\draw[fill=black] (1-5,1-6) circle [x radius=0.1333, y radius=0.1333, rotate=0];
\draw[fill=white] (0-5,0-6) circle [x radius=0.4376, y radius=0.4376, rotate=0]; 

\draw[fill=black] (1-5+4,-1-6) circle [x radius= 0.0553, y radius= 0.0553, rotate=0]; 
\draw[red,fill=white] (0-5+4,0-6) circle [x radius=0.0065, y radius=0.0065, rotate=0]; 

\draw[fill=white] (0-5+4+4,0-6) circle [x radius= 0.1135, y radius= 0.1135, rotate=0]; 
\draw[fill=black] (1-5+4+4,-1-6) circle [x radius=0.4259, y radius=0.4259, rotate=0]; 


\end{tikzpicture}

};

\node (3) at (10+0.07,-4-0.05) {

\begin{tikzpicture}[scale=0.7]

\node (0) at (0-5,0-6) {};
\node (1) at (1-5,1-6) {};
\node (2) at (1-5,-1-6) {$\times$};
\node (3) at (2-5,0-6) {$\times$};


\draw[dashed] (0) -- (1);
\draw[dashed] (1)--(3);
\draw (3)--(2);
\draw (2)--(0);

\node (4) at (4-5,0-6) {};
\node (5) at (5-5,1-6) {$\times$};
\node (6) at (5-5,-1-6) {};
\node (7) at (6-5,0-6) {$\times$};


\draw[dashed] (6)--(4);
\draw (4)--(5);
\draw[dashed] (5)--(7);
\draw (7)--(6);

\node (8) at (8-5,0-6) {};
\node (9) at (9-5,1-6) {$\times$};
\node (10) at (9-5,-1-6) {};
\node (11) at (10-5,0-6) {$\times$};


\draw[dashed] (8)--(10)--(11);
\draw (11)--(9)--(8);

\draw (1)--(4);
\draw (4)--(2);

\draw (7)--(9); 
\draw[dashed] (7)--(10);




\draw[fill=white] (0-5,0-6) circle [x radius=0.4223, y radius=0.4223, rotate=0]; 
\draw[fill=black] (1-5,1-6) circle [x radius=0.1505, y radius=0.1505, rotate=0]; 

\draw[red,fill=white] (0-5+4,0-6) circle [x radius=0.0070, y radius=0.0070, rotate=0]; 
\draw[fill=black] (1-5+4,-1-6) circle [x radius= 0.0673, y radius= 0.0673, rotate=0]; 

\draw[fill=white] (0-5+4+4,0-6) circle [x radius=0.4305, y radius=0.4305, rotate=0]; 
\draw[fill=black] (1-5+4+4,-1-6) circle [x radius=0.1681, y radius=0.1681, rotate=0]; 


\end{tikzpicture}
};

\node (4) at (10+0.03,-2+0.3+0.1) {\begin{tikzpicture}[scale=0.7]

\node (0) at (0-5,0-6) {};
\node (1) at (1-5,1-6) {};
\node (2) at (1-5,-1-6) {$\times$};
\node (3) at (2-5,0-6) {$\times$};


\draw[dashed] (0) -- (1);
\draw[dashed] (1)--(3);
\draw (3)--(2);
\draw (2)--(0);

\node (4) at (4-5,0-6) {};
\node (5) at (5-5,1-6) {$\times$};
\node (6) at (5-5,-1-6) {};
\node (7) at (6-5,0-6) {$\times$};


\draw[dashed] (6)--(4);
\draw (4)--(5);
\draw[dashed] (5)--(7);
\draw (7)--(6);

\node (8) at (8-5,0-6) {};
\node (9) at (9-5,1-6) {$\times$};
\node (10) at (9-5,-1-6) {};
\node (11) at (10-5,0-6) {$\times$};


\draw[dashed] (8)--(10)--(11);
\draw (11)--(9)--(8);

\draw[dashed] (1)--(4);
\draw (4)--(2);

\draw (7)--(9); 
\draw[dashed] (7)--(10);




\draw[fill=white] (1-5,1-6) circle [x radius=0.4323, y radius=0.4323, rotate=0]; 
\draw[fill=black] (0-5,0-6) circle [x radius=0.1642, y radius=0.1642, rotate=0]; 

\draw[red,fill=black] (1-5+4,-1-6) circle [x radius=0.0242, y radius=0.0242, rotate=0]; 
\draw[red,fill=white] (0-5+4,0-6) circle [x radius=0.0073, y radius=0.0073, rotate=0]; 

\draw[fill=white] (0-5+4+4,0-6) circle [x radius=0.4359, y radius=0.4359, rotate=0]; 
\draw[fill=black] (1-5+4+4,-1-6) circle [x radius=0.1675, y radius=0.1675, rotate=0]; 


\end{tikzpicture}

};

\node (5) at (10+0.02,0+0.6+0.1) {

\begin{tikzpicture}[scale=0.7]

\node (0) at (0-5,0-6) {};
\node (1) at (1-5,1-6) {};
\node (2) at (1-5,-1-6) {$\times$};
\node (3) at (2-5,0-6) {$\times$};


\draw[dashed] (0) -- (1);
\draw[dashed] (1)--(3);
\draw (3)--(2);
\draw (2)--(0);

\node (4) at (4-5,0-6) {$\times$};
\node (5) at (5-5,1-6) {$\times$};
\node (6) at (5-5,-1-6) {};
\node (7) at (6-5,0-6) {};


\draw (6)--(4)--(5);
\draw[dashed] (5)--(7)--(6);

\node (8) at (8-5,0-6) {};
\node (9) at (9-5,1-6) {$\times$};
\node (10) at (9-5,-1-6) {};
\node (11) at (10-5,0-6) {$\times$};


\draw[dashed] (8)--(10)--(11);
\draw (11)--(9)--(8);

\draw[dashed] (1)--(4);
\draw (4)--(2);

\draw (7)--(9); 
\draw[dashed] (7)--(10);




\draw[fill=white] (0-5,0-6) circle [x radius=0.496, y radius=0.496, rotate=0]; 
\draw[fill=black] (1-5,1-6) circle [x radius=0.1702, y radius=0.1702, rotate=0]; 

\draw[fill=white] (2-5+4,0-6) circle [x radius=0.0919, y radius=0.0919, rotate=0]; 
\draw[fill=black] (1-5+4,-1-6) circle [x radius=0.4268, y radius=0.4268, rotate=0]; 

\draw[fill=white] (0-5+4+4,0-6) circle [x radius=0.4514, y radius=0.4514, rotate=0]; 
\draw[fill=black] (1-5+4+4,-1-6) circle [x radius=0.1689, y radius=0.1689, rotate=0]; 


\end{tikzpicture}
};


\end{tikzpicture}
\caption{The local density of states, calculated as an integral of the experimentally measured impedance spectra over an energy window of $\varepsilon\in [-0.102,0.102]$ with a corresponding frequency range of 111 MHz to 126 MHz. The range was chosen to account for shifts in LDOS maxima, arising from variance in cable length and from T-connectors. 50$\Omega$ RG58 cables are indicated with a dashed line, and 93$\Omega$ RG62 cables with a solid line. Two sites were measured on each section, with an $\times$ depicting unmeasured sites. On measured sites the diameter of the circle is proportional to the calculated integral, rescaled to the maximum measured value out of the entire sequence (the white site on $g_2$ of the final structure). To distinguish a true zero energy state from broadened low lying states, red indicates if the LDOS at $\varepsilon=0$ had a local minima in the experimentally measured LDOS. The arrows denote the position of the structure to the corresponding classification in the sequence \eqref{Seq1}. On a topological phase boundary, a section that has support of a zero energy state on both a black site and a white site indicates two zero energy states originate from that section allowing the number of zero energy states, singular sections, and if the structure is on a higher steps phase boundary to be inferred. Reading off the figure, (b), (c), (f) are on a topological phase boundary, confirming the predicted sequential classification. The non-zero value calculated in (a) is a consequence of broadened peaks of low lying states in the measured spectra, with a local minima found at $\varepsilon=0$ in experiment.
}
\label{LDOSexpsresults_OtherSequence}
\end{figure}
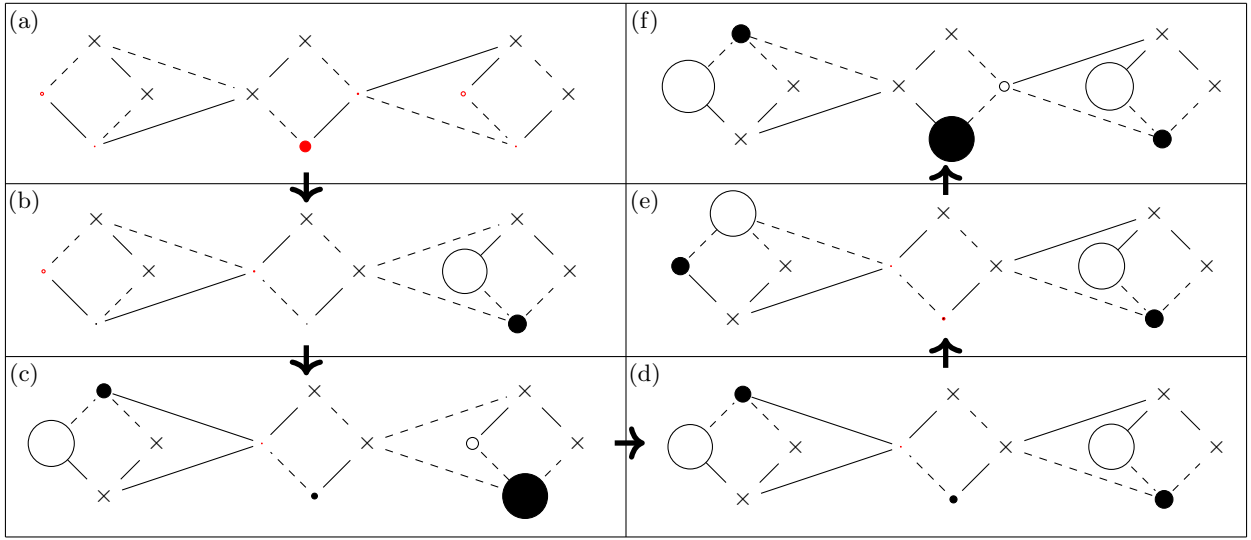

\end{document}